\RequirePackage{ifthen}

\newcommand{\typeof}{1} %
\newcommand{\longversion}[1]{\ifthenelse{\equal{\typeof}{0}}{}{#1}}
\newcommand{\shortversion}[1]{\ifthenelse{\equal{\typeof}{0}}{#1}{}}
\newcommand{\longshortversion}[2]{\ifthenelse{\equal{\typeof}{0}}{#2}{#1}}

\newcommand{\highlight}[1]{{{#1}}}

\longshortversion{
	\documentclass[a4paper,10pt]{article}
	
	\title{Contextual Behavioural Metrics (Extended 
	Version)}
	\author{Ugo Dal Lago \and Maurizio Murgia}
	\date{}
	
	\usepackage{a4wide}
	\usepackage{amsthm}
	\theoremstyle{plain}
	\newtheorem{theorem}{Theorem}[section]
	\newtheorem{proposition}[theorem]{Proposition}
	\newtheorem{lemma}[theorem]{Lemma}
	
	\theoremstyle{definition}
	\newtheorem{definition}[theorem]{Definition}
	\newtheorem{example}{Example}
	\newtheorem{remark}{Remark}
	
	\usepackage{amsmath}
	\usepackage{amsfonts}
	\usepackage{amssymb}
	\usepackage{color}
	\usepackage{proof}
	\usepackage{mathrsfs}  
	\usepackage{bm}
	\usepackage{stmaryrd}
	\usepackage{mathtools}
	\usepackage{bbm}
	\usepackage{tikz}\usetikzlibrary{quantikz,automata, arrows.meta, positioning}
	\usepackage{forest}
	\usepackage{wrapfig}
	\usepackage{empheq}
	\usepackage{url}
	\usepackage[inference]{semantic}\setpremisesend{0pt}\setpremisesspace{10pt}\setnamespace{0pt}
	\usepackage{csquotes}
	\usepackage[strings]{underscore}
	\usepackage{eucal}
	\usepackage{subcaption}
	\usepackage{url}
	\usepackage{bussproofs}
	\usepackage{hyperref} 
	\usepackage{cleveref}
	
	\bibliographystyle{plain}	
	
\newcommand{\subst}[2]{\{n\leftarrow p\}}

\newcommand{\quantaleQ}{\mathbb{Q}}
\newcommand{\quantaleB}{\mathbb{B}}
\newcommand{\automatonV}{\mathbb{V}}
\newcommand{\mbot}[1][]{\bot_{#1}}
\newcommand{\mtop}[1][]{\top_{#1}}
\newcommand{\mplus}[1][]{+_{#1}}
\newcommand{\lub}[1][]{\bigvee_{#1}}

\newcommand{\glb}[1][]{\bigwedge_{#1}}

\newcommand{\lessThan}[1][]{\leq_{#1}}
\newcommand{\lessThanStrict}[1][]{<_{#1}}
\newcommand{\greatThan}[1][]{\geq_{#1}}
\newcommand{\greatThanStrict}[1][]{>_{#1}}
\newcommand{\mequiv}[1][]{\leq\geq_{#1}}
\newcommand{\elE}{e}
\newcommand{\elEi}{e'}
\newcommand{\elF}{f}

\newcommand{\proc}{P}
\newcommand{\labels}{\mathcal{L}}

\newcommand{\states}{S}
\newcommand{\statesi}{S'}
\newcommand{\statesii}{S''}

\newcommand{\trans}[1][]{\xrightarrow{#1}}
\newcommand{\stateS}[1][]{s_{#1}}
\newcommand{\stateSi}[1][]{s'_{#1}}
\newcommand{\stateSii}[1][]{s''_{#1}}
\newcommand{\stateSiii}[1][]{s'''_{#1}}
\newcommand{\val}{\Downarrow}
\newcommand{\labelL}[1][]{\ell_{#1}}
\newcommand{\labelLi}[1][]{\ell'_{#1}}
\newcommand{\labelA}{a}
\newcommand{\labelB}{b}

\newcommand{\procP}[1][]{p_{#1}}
\newcommand{\procPi}[1][]{p'_{#1}}

\newcommand{\procQ}[1][]{q_{#1}}
\newcommand{\procQi}[1][]{q'_{#1}}

\newcommand{\procR}[1][]{r_{#1}}

\newcommand{\env}{E}

\newcommand{\setcomp}[2]{\{#1\;|\;#2\}}
\newcommand{\relR}{\mathcal{R}}

\newcommand{\metric}[1][]{\delta_{#1}}

\newcommand{\immMetric}{D}
\newcommand{\cbmap}[1][]{m_{#1}}

\newcommand{\paral}{|}
\newcommand{\esum}{+}
\newcommand{\new}{\nu}
\newcommand{\pref}{.}
\newcommand{\bang}{!}
\newcommand{\irule}[3]{\frac{#1}{#2}{\makebox[1pt][l]{#3}}}
\newcommand{\similar}{\precsim}
\newcommand{\revsimilar}{\succsim}
\newcommand{\mutsimilar}{\similar\revsimilar}

}{
	\documentclass[a4paper,UKenglish,cleveref, autoref, thm-restate]{lipics-v2021}
	\bibliographystyle{plainurl}% the mandatory bibstyle
	
	\title{Contextual Behavioural Metrics}
	\author{Ugo {Dal Lago}}{University of Bologna, Italy \and INRIA Sophia Antipolis, France}{ugo.dallago@unibo.it}{https://orcid.org/0000-0001-9200-070X}{}
	
	\author{Maurizio Murgia}{Gran Sasso Science Institute, L'Aquila, 
	Italy}{maurizio.murgia@gssi.it}{https://orcid.org/0000-0001-7613-621X}{}	

	\authorrunning{Ugo Dal Lago and Maurizio Murgia}
	
	\Copyright{Ugo Dal Lago and Maurizio Murgia} %TODO mandatory, please use 
	 %TODO mandatory; please add 
	\funding{The first author is supported by the ERC CoG “DIAPASoN”, GA 818616, and 
	by PNRR - M4C2 - I1.4 Project CN00000013 "National Centre for HPC, Big Data and Quantum 
	Computing”. The second author is supported by MUR project PON REACT EU DM 1062/21.}
	
	\acknowledgements{We thank the anonymous reviewers for their feedback.}
	
	\EventEditors{Guillermo A. P\'{e}rez and Jean-Fran\c{c}ois Raskin}
	\EventNoEds{2}
	\EventLongTitle{34th International Conference on Concurrency Theory (CONCUR 2023)}
	\EventShortTitle{CONCUR 2023}
	\EventAcronym{CONCUR}
	\EventYear{2023}
	\EventDate{September 18--23, 2023}
	\EventLocation{Antwerp, Belgium}
	\EventLogo{}
	\SeriesVolume{279}
	\ArticleNo{38}
	\usepackage{amsmath}
	\usepackage{amsfonts}
	\usepackage{amssymb}
	\usepackage{color}
	\usepackage{proof}
	\usepackage{mathrsfs}  
	\usepackage{bm}
	\usepackage{stmaryrd}
	\usepackage{mathtools}
	\usepackage{bbm}
	\usepackage{tikz}\usetikzlibrary{quantikz,automata, arrows.meta, positioning}
	\usepackage{forest}
	\usepackage{wrapfig}
	\usepackage{empheq}
	\usepackage{url}
	\usepackage[inference]{semantic}\setpremisesend{0pt}\setpremisesspace{10pt}\setnamespace{0pt}
	\usepackage{csquotes}
	\usepackage[strings]{underscore}
	\usepackage{eucal}
	\usepackage{subcaption}
	\usepackage{bussproofs}
	
	\nolinenumbers

}

\newcommand{\commentout}[1]{}

\begin{document}

\maketitle

\begin{abstract}
We introduce contextual behavioural metrics (CBMs) as a novel way of measuring 
the discrepancy in behaviour between processes, taking into account both 
quantitative aspects and contextual information. 
This way, process distances by construction take the environment into account: 
two (non-equivalent) processes may still exhibit very similar behaviour in some 
contexts, e.g., when certain actions are never performed. We first show how 
CBMs capture many well-known notions of equivalence and metric, including 
Larsen's environmental parametrized bisimulation. We then study compositional 
properties of CBMs with respect to some common process algebraic operators, 
namely prefixing, restriction, non-deterministic sum, parallel composition and 
replication.
\end{abstract}

\section{Introduction}

Simulation and bisimulation relations are often the methodology of choice for 
reasoning relationally about the behaviour of systems specified in the form of 
LTSs. On the one hand, most of them can be proved to be congruences, therefore 
enabling modular equivalence proofs. On the other hand, not being based on any
universal quantification (e.g. on tests or on traces), they enable 
simpler relational arguments, especially when combined with enhancements such 
as the so-called up-to techniques~\cite{PousSangiorgi2011}.

The outcome of relational reasoning as supported by (bi)simulation relations is 
inherently binary: two programs or systems are \emph{either} (bi)similar 
\emph{or not so}. As an example, all pairs of 
non-equivalent elements have the same status, i.e. the bisimulation game
gives no information on the degree of dissimilarity between 
non-equivalent states. This can be a problem in those contexts, such as 
that of probabilistic systems, in which non-equivalent states can give rise to 
completely different but also extremely similar behaviours.

This led to the introduction of a generalization of bisimulation relations, 
i.e. the so-called bisimulation \emph{metrics}~\cite{DGJP1999}, which 
rather than being 
binary relations on the underlying set of states $S$, are binary \emph{maps} 
from $S$ to a quantale (most often of real numbers) satisfying the axioms 
of (pseudo)metrics. In that context, the bisimulation game becomes inherently quantitative: the 
defender aims at proving that the two states at hand are \emph{close} to each 
other, while the attacker tries to prove that they are \emph{far apart}. The 
outcome of this game is a quantity 
representing a bound not only on any discrepancy about the \emph{immediate} 
behaviour 
of the two involved states, 
(e.g. the fact that some action is available in $s$ but not in $t$), but 
also providing some information about differences which will only show up in the 
\emph{future}, all this \emph{regardless} of the actions chosen by the 
attacker. In this sense, therefore, bisimulation metrics condense a great deal 
of information in just one number. 

Notions of bisimulation metrics have indeed be defined for various sequential 
and concurrent calculi (see, 
e.g.,~\cite{ChatzikokolakisGPX14,DJGP2002,DuDG16,FPP2004,GeblerLT16,BHMW2005}), 
allowing a form of 
metric reasoning on program 
behaviour. But when could any of such techniques be said to be compositional? 
This amounts to be able to \emph{derive} an upper bound on the distance 
$\delta(C[t],C[s])$ between two programs in the form $C[t]$ and $C[s]$ 
\emph{from} the distance $\delta(s,t)$ between $s$ and $t$. Typically, the 
latter is required to be itself an upper bound on the former, giving rise to 
\emph{non-expansiveness} as a possible generalization of the notion of a 
congruence. This, however, significantly \emph{restricts} the class of 
environments $C$ to which the aforementioned analysis can be applied, since 
being able to amplify differences is a very natural property of processes. 
Indeed, an inherent tension exists between expressiveness and 
compositionality in metric reasoning~\cite{GeblerLT16}.

But there is another reason why behavioural metrics can be seen as less 
informative 
than they could be. As already mentioned, any  
number measuring the distance between two states $s$ and $t$ implicitly
accounts for all the possible ways of comparing $s$ and $t$, i.e. any 
context. Often, however, only contexts that act in a certain very specific way could 
highlight large differences between $s$ and $t$, while others might simply see 
$s$ and $t$ as very similar, or even equivalent. This further dimension is 
abstracted away in compositional metric analysis: if the distance between $s$ 
and $t$ is very high, but $C$ does not ``take advantage'' of such large 
differences, $C[s]$ and $C[t]$ should be \emph{close} to each other, but are 
dubbed being \emph{far away} from each other, due to the aforementioned 
abstraction step. It is thus natural to wonder whether metric analysis can be 
made contextual. In the realm of process equivalences, this is known to be 
possible through, e.g. Larsen's environmental parametrized 
bisimulation~\cite{Larsen87}, but not much is known about contextual 
enhancements of bisimulation \emph{metrics}. Other notions of program 
equivalence, like logical relations or denotational semantics, have been shown 
to have metric analogues~\cite{AGHKC17,ReedPierce2010}, which in some cases can 
be made contextual~\cite{GeoffroyPistone2021,LagoGY19}.

In this paper, we introduce the novel notion of \emph{contextual behavioural 
metric} (CBM 
in the following) through which it is possible to fine-tune the abstraction 
step mentioned above and which thus represents a refinement over behavioural 
metrics. In CBMs, the distance between two states $s,t$ of an LTS is measured 
by an object $d$ having a richer structure than that of a number. 
Specifically, $d$ is taken to be an element of a metric
transition system, in which the contextual and temporal dimensions of the 
differences can be taken into account. In addition to the mere introduction of this new notion of distance, our contributions are threefold:
\begin{itemize}
	\item
	On the one hand, we show that metric labelled transition systems (MLTSs in 
	the following), namely the kind of structures meant to model differences, 
	are indeed quantales, this way allowing us to prove that CBMs are 
	generalized metrics. This is in 
	Section~\ref{sect:CBMs}.
	\item
	On the other hand, we prove that some well-known methodologies for 
	qualitative and 
	quantitative relational reasoning on processes, namely (strong) 
	bisimulation 
	relations and metrics, and environmental parametrized 
	bisimulations~\cite{Larsen87}, can all be 
	seen as CBMs where the underlying MLTS corresponds to the 
	original quantale. This is in Section~\ref{sect:Examples}.
	\item
	Finally, we prove that CBMs have some interesting 
	compositional properties, and that this allows one to derive approximations 
	to the distance between processes following their syntactic structure. This 
	is in Section \ref{sect:Composition}.
\end{itemize}

Many of the aforementioned works about behavioural metrics are concerned with 
probabilistic forms of LTSs. In this work, instead, we have deliberately chosen 
to focus on usual nondeterministic transition systems. On the one hand, the 
quantitative aspects can be handled through the so-called \emph{immediate} 
distance between states, see below. On the other hand, it is well known that 
probabilistic transition systems can be seen as (non)deterministic systems  
whose underlying reduction relation is defined between state distributions. 
Focusing on ordinary LTSs has the advantage of allowing us to 
concentrate our attention on those aspects related to metrics, allowing for a 
separation of concerns. This being said, we are confident that most of the results 
described here could hold for probabilistic LTSs, too.

%\shortversion{\noindent Proofs and other details are omitted due to space constraints, and
%can be found online \cite{EV}.}

\section{Why the Environment Matters}\label{sect:Why}
The purpose of this section is to explain why purely numerical quantales do 
\emph{not} precisely capture differences between states of an LTS and how 
a more structured approach to distances can be helpful to tackle this 
problem. We will do this through an example drawn from the realm of 
higher-order programs, the latter seen as states of the LTS induced by 
Abramsky's applicative bisimilarity~\cite{Abramsky90}.

Let us start with a pair of programs written in a typed 
$\lambda$-calculus, both of them having type 
$(\mathtt{Nat}\rightarrow\mathtt{Nat})\rightarrow\mathtt{Nat}$, namely
$M_2$ and $M_4$, where $M_n\triangleq\lambda x.xn$. These terms can indeed be 
seen as states of an LTS, whose relevant fragment is the following one:
\vspace{-0.15cm}
\begin{center}
	\includegraphics{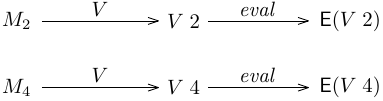}
\end{center}
\vspace{-0.3cm}
Labelled transitions correspond to either parameter passing (each actual 
parameter being captured by a distinct label $V$) or evaluation. It 
is indeed convenient to see the underlying LTS as a bipartite structure 
whose states are either computations or values. The two states 
$\mathsf{E}(V\;2)$ and $\mathsf{E}(V\;4)$ are the natural number values to 
which $V\;2$ and $V\;4$ evaluate, respectively. Clearly, the latter are not to 
be considered equivalent whenever different, and this can be captured, e.g., by 
either exposing the underlying numerical value through a labelled 
self-transition or 
by stipulating that base type values, contrary to higher-order values, can 
be explicitly observed, thus being equivalent precisely when equal. If one 
plays the 
bisimulation game on top of this LTS, the resulting 
notion of equivalence turns out to be precisely Abramsky's applicative 
bisimilarity. For very good reasons, $M_2$ and $M_4$ are dubbed as \emph{not} 
equivalent: they can be separated by feeding, e.g. $V=\lambda x.x$ to them.

But now, \emph{how far apart} should $M_2$ and $M_4$ be?  The answer provided 
by behavioural metrics consists in saying that $M_2$ and $M_4$ are at distance 
\emph{at most} $x\in\mathbb{R}^\infty_+$ iff $x$ is an upper bound on the 
differences any adversary observes while interacting with them, 
\emph{independently} on 
how the adversary behaves. As a consequence, if the underlying 
$\lambda$-calculus provides a primitive for multiplication, then it is indeed 
possible to define values of the form $V_n\triangleq\lambda x.x\times n$ for 
every $n$, allowing the environment to observe arbitrarily large differences of
the form 
$$
\mid\mathsf{E}(V_n\;2)-\mathsf{E}(V_n\;4)\mid\;=\;\mid 2n-4n\mid\;=2n.
$$
In other words, the distance between $M_2$ and $M_4$ is $+\infty$. The 
possibility of arbitrarily amplifying distances is well-known, and 
can be tackled, e.g., by switching to a calculus in which \emph{all} 
functions are \emph{non-expansive}, ruling out terms such as $V_n$ where $n>1$. 
In other words, the distance between $M_2$ and $M_4$ is 
indeed $2$, because no input term $V$ can ``stretch'' the distance 
between $2$ to $4$ to anything more than $2$. This is what happens, e.g., in 
$\mathsf{Fuzz}$~\cite{ReedPierce2010}.

But is this the end of the story? Are we somehow losing too much information 
by stipulating that $M_2$ and $M_4$ are, say, at distance $2$? Actually, the 
only moment in which the environment observes the state with which it is 
interacting is at the end of the dialogue, namely after feeding it with a 
function $V:\mathtt{Nat}\rightarrow\mathtt{Nat}$. If, for example, the 
environment picks $V_q\triangleq\lambda x.(x-3)^2+2$, then the observed 
difference is $0$, while if it picks $V_l\triangleq\lambda x.x+2$ then the 
observed distance is maximal, i.e. $2$. In other words, the observed distance 
strictly depends on how the environment behaves and should arguably be 
parametrised on it. This is indeed the main idea 
behind Larsen's environmental parametrised bisimulation, but also behind our
contextual behavioural metrics. In the latter, differences can be faithfully 
captured by the states of \emph{another} labelled transition system, called a 
metric labelled transition system, in which 
observed distances are associated to states. In our example, the difference 
between $M_2$ and $M_4$ \emph{is} the state $s$ of a metric labeled 
transition system whose relevant fragment is:
\vspace{-0.3cm}
\begin{center}
	\includegraphics{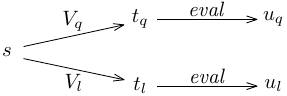}
\end{center}
\vspace{-0.3cm}
Crucially, while $s,t_s,t_l,u_l$ are all mapped to the null observable 
difference, $u_q$ is associated to $2$. This allows to discriminate between 
those environments which are able to see large differences from those which are 
not. This is achieved by allowing differences to be modelled by the states of a 
transition system \emph{themselves}. Using a categorical jargon, it looks 
potentially useful, but also very 
tempting, to impose the structure of a 
coalgebra to the underlying space of distances rather than taking it as a 
monolithical, numeric, quantale. The rest of this paper can be seen as an 
attempt to make this idea formal.

\section{Contextual Behavioural Metrics, Formally}\label{sect:CBMs}
%!TeX root=main.tex
%\subsection{Preliminaries}

%unital integral commutative quantale.
This section is devoted to introducing contextual behavioural metrics, namely 
the concept we aim at studying in this paper. We start with the definition of 
quantale \cite{quantalesBook}, the canonical codomain of generalized metrics 
\cite{Lawvere73}. The notion of quantale used in this paper is that of unital 
integral commutative quantale:

\begin{definition}[Quantale]
A \emph{quantale} is a structure 
$\quantaleQ = (Q,\glb,\lub,\mbot,\mtop,\mplus)$ such
that $\glb,\lub: 2^Q \to Q$, the two objects $\mbot,\mtop$ are in $Q$, and 
$\mplus$ is a binary operation on $Q$, where:
\begin{itemize}
\item $(Q,\glb,\lub,\mbot,\mtop)$ is a complete lattice;
\item $(Q,\mplus,\mbot)$ is a commutative monoid;
\item for every $\elE \in Q$ and every $A \subseteq Q$ it holds that $\elE 
\mplus \glb A = \glb \setcomp{\elE \mplus \elF}{\elF \in A}$.
\end{itemize}
We write $\elE \lessThan \elF$ when $\elE = \glb \{\elE,\elF\}$.
\end{definition}

Generalized metrics are maps which associate an element of a given quantale to 
each pair of elements. As customary in behavioural metrics, we work with 
\emph{pseudo}metrics, in which distinct elements may be at minimal 
distance:
\begin{definition}[Metrics]
A \emph{pseudometric} over a set $A$ with values in a quantale $\quantaleQ$ is 
a map $m: A \times A \to \quantaleQ$ satisfying:
\begin{itemize}
\item for all $a \in A: m(a,a) = \mbot$;
\item for all $a,b \in A: m(a,b) = m(b,a)$;
\item for all $a,b,c \in A: m(a,c) \lessThan m(a,b) \mplus m(b,c)$.
\end{itemize}
In the rest of this paper, we refer to pseudometrics simply as metrics.
\end{definition}

It is now time to introduce our notion of a \emph{process}, namely of the 
computational objects we want to compare. We do not fix a syntax, and work 
with abstract labelled transition systems (LTSs in the following). In order to 
enable 
(possibly quantitative) metric reasoning, we equip states of our LTS with an 
immediate metric $\immMetric$, namely a metric measuring the observable 
distance between two states.
\begin{definition}[Process LTS]
We define a $\quantaleQ$-LTS as a quadruple $(\proc,\labels,\trans,\immMetric)$ where:
\begin{itemize}
\item $\proc$ is the set of processes;
\item $\labels$ is the set of labels;
\item $\trans\;\subseteq\;\proc \times \labels \times \proc$ is the transition 
relation;
\item $\immMetric: \proc \times \proc \to \quantaleQ$ is a metric.
\end{itemize}
%(ranged over by $\procP,\procQ,\hdots$)
\end{definition}
\begin{example}
The example LTS from Section~\ref{sect:Why} should be helpful in understanding 
why the metric $\immMetric$ is needed: terms and values of distinct types are 
at maximal immediate distance, while terms and values of the same type are at 
minimal distance, except when the type is $\mathtt{Nat}$, whereas the immediate 
distance is just the absolute value between the two numbers.
\end{example}

We now need to introduce \emph{another} notion of transition system, this time meant to 
model \emph{differences} between computations. This kind of structure can be 
interpreted as a quantale, and will form the codomain of Contextual 
Bisimulation Metrics. Intuitively, a Metric LTS is an LTS endowed with a 
function from states to a quantale $\quantaleQ$. This allows to keep track of 
immediate distance changes. Let us start with the notion of a \emph{pre}-metric 
LTS:
\begin{definition}[Pre-metric LTS]\label{def:PMLTS}
A \emph{pre-metric $\quantaleQ$-LTS} is a quadruple 
$\automatonV=(\states,\labels,\trans,\val)$
where:
\begin{itemize}
\item $\states$ is the set of states;
\item $\labels$ is the set of labels;
\item $\trans\;\subseteq\;\states \times \labels \times \states$ is the 
transition relation;
\item $\val\;:\states \to \quantaleQ$ is a function which assigns values 
in $\quantaleQ$ to states in $\states$.
\end{itemize}
\end{definition}

A pre-metric LTS does not necessarily form a quantale, because $\states$ does not
necessarily have, e.g. the structure of a monoid or a lattice. In order to 
be proper codomains for metrics, pre-metric LTSs need to be endowed with some 
additional structure, which will be proved to be enough to form a quantale. 
\begin{definition}[Metric LTS]\label{def:MLTS}
A metric $\quantaleQ$-LTS $\automatonV = (\states,\labels,\trans,\val)$ is a 
pre-metric $\quantaleQ$-LTS endowed with two elements
$\mbot[\automatonV],\mtop[\automatonV] \in \states$, and three
operators
$\glb[\automatonV],\lub[\automatonV]:2^{\states} \to \states$ and 
$\mplus[\automatonV]: \states \times \states \to \states$, 
where the conditions hold for all possible values of the involved 
metavariables:
%for all $\labelL \in \labels, \statesi 
%\subseteq \states, \stateS[1],\stateS[2] \in 
%\states$:
\[\arraycolsep=15pt
\begin{array}{cc}
\mbot[\automatonV] \trans[\labelL] \stateS \iff
\stateS = \mbot[\automatonV]&
 \val \mbot[\automatonV] = \mbot[\quantaleQ]\\[5pt]
\forall \labelL \in \labels: \mtop[\automatonV] \not\trans[\labelL]  &
 \val \mtop[\automatonV] = \mtop[\quantaleQ]\\[5pt]
\glb[\automatonV] \statesi \trans[\labelL] \stateS \iff \exists \stateSi \in \statesi:
\stateSi \trans[\labelL] \stateS&\val \glb[\automatonV] \statesi = \glb[\quantaleQ] 
\setcomp{\val \stateS}{\stateS \in \statesi}\\[5pt]
\lub[\automatonV] \statesi \trans[\labelL] \stateS \iff \exists \statesii:\stateS = \lub[\automatonV] \statesii\;\text{and}\;
& \val \lub[\automatonV] \statesi = \lub[\quantaleQ] \setcomp{\val \stateS}{\stateS \in \statesi}\\
\exists\;\text{surjective}\;f: \statesi \to \statesii: \forall \stateSi \in \statesi:  \stateSi \trans[\labelL] f(\stateSi)\\[5pt]
%\forall \stateSi \in \statesi: \exists \stateSii \in \statesii: \stateSi \trans[\labelL] \stateSii
%\;\text{and}\;
%\forall \stateSii \in \statesii: \exists \stateSi \in \statesi: \stateSi \trans[\labelL] \stateSii
\stateS[1] \mplus[\automatonV] \stateS[2] \trans[\labelL] \stateSi
 \iff \stateSi = \stateSi[1] \mplus[\automatonV] \stateSi[2]\;\text{for some}\;\stateSi[1],\stateSi[2]
 &
\val{(\stateS[1] \mplus[\automatonV] \stateS[2])} = \val{\stateS[1]} 
\mplus[\quantaleQ] \val{\stateS[2]}\\
 \text{such that: }
\stateS[1] \trans[\labelL] \stateSi[1] \;\text{and}\; \stateS[2] \trans[\labelL] \stateSi[2]
%\transF(\stateS[1] \mplus[\automatonV] \stateS[2],\labelL) = \transF(\stateS[1],\labelL) \mplus[\automatonV] \transF(\stateS[2],\labelL) &
% \val (\stateS[1] \mplus[\automatonV] \stateS[2]) = (\val\stateS[1]) \mplus[\quantaleQ] (\val\stateS[2])\\
\end{array}
\]
%We do not define $\mplus[\automatonV]$ as it needs to be specialised to the particular operator under consideration. We only require that:
%\begin{itemize}
%\item $\forall \stateS[1],\stateS[2]:
%\val (\stateS[1] \mplus[\automatonV] \stateS[2]) = (\val\stateS[1]) 
%\mplus[\quantaleQ] (\val\stateS[2])$;
%\item $\forall \stateS,\forall \statesi \subseteq \states:
%\stateS \mplus[\automatonV] \glb \statesi \mequiv[\automatonV] \glb \setcomp{\stateS 
%\mplus[\automatonV] \stateSi}{\stateSi \in \statesi}$.
%\end{itemize}
\end{definition}
\highlight{Axioms ensures that $\mbot[\automatonV]$ allows every possible behaviour (somehow 
capturing every context), and dually $\mtop[\automatonV]$ disallows every behaviour.
$\glb[\automatonV]\statesi$ allows all and only the behaviours in $\statesi$ 
(union of contexts), while $\lub[\automatonV] \statesi$ enables all and only the
behaviours allowed by \emph{every} element in $\statesi$ (intersection of contexts). 
The sum $\mplus[\automatonV]$ has a behaviour similar to $\lub[\automatonV]$, but it is
binary and differs on the value returned by $\val$.}

\begin{remark}[On The Existance Of Non-Trivial MLTSs]
Due to the requirements about joins and meets over potentially infinite sets, MLTSs are not easy to define directly. We argue, however, that an MLTS can be defined as the closure of a pre-MLTS. If the underlying quantale $\quantaleQ$ is boolean, one can get the desired structure by considering $2^{2^{X}}$, where $X$ is the carrier of the given pre-MLTS: it suffices to take subsets in ``conjunctive'' normal form. For the general case, the class $\cup_{n\in\mathbb{N}}{\overbrace{2^{\hdots^{2^{X}}}}^{n\;\text{times}}}$, which is indeed a set in ZFC, suffices.
\end{remark}

The axiomatics above is still not sufficient to give the status of a quantale 
to $\quantaleQ$-MLTSs. The reason behind all this is that there could be 
equivalent but distinct states in $\states$. We then define a preorder 
$\lessThan[\automatonV]$ on the states of any MLTS $\automatonV$:  
\begin{definition}
A relation $\relR \subseteq \states \times \states$ is a 
\emph{$\lessThan[\quantaleQ]$-preserving simulation}\footnote{Technically, it 
is a reverse simulation. We call it simulation for brevity.} 
if, whenever $\stateS[1]\; \relR\; \stateS[2]$, it holds that:
\begin{enumerate}
\item $\val 
\stateS[1]\;\lessThan[\quantaleQ]\;\val\stateS[2]$;\label{def:lessThan:item1}
\item  $\forall \labelL \in \labels: \stateS[2] \trans[\labelL] \stateSi[2] \implies
\exists \stateSi[1]: \stateS[1] \trans[\labelL] \stateSi[1] 
\;\text{and}\;\stateSi[1] \relR\stateSi[2]$.\label{def:lessThan:item2}
\end{enumerate}
We define $\lessThan[\automatonV] \subseteq \states \times \states$ as the largest 
$\lessThan[\quantaleQ]$-preserving simulation. We use the notation 
$\mequiv[\automatonV]$ for mutual $\lessThan[\quantaleQ]$-preserving simulation, that is $\mequiv[\automatonV] = \lessThan[\automatonV] \cap \greatThan[\automatonV]$.
We say that $\stateS$ is a lower (resp. upper) bound of $\statesi \subseteq \states$ 
if $\stateS \lessThan[\automatonV] \stateSi$ 
(resp. $\stateSi \lessThan[\automatonV] \stateS$) for all $\stateSi \in \statesi$.
\end{definition}

The forthcoming result states that, in general, MLTSs \emph{almost} form 
quantales. We can recover a proper quantale by quotienting $\states$ modulo 
$\mequiv[\automatonV]$.
  
\begin{proposition}[Properties of MLTSs]\label{lem:mlts-properties}
Let $\automatonV = (\states,\labels,\trans,\val)$ be a MLTS. 
Then:
\begin{enumerate}
\item $\lessThan[\automatonV]$ is a preorder relation;
\item For all $\stateS$: $\mbot[\automatonV] \lessThan[\automatonV] \stateS$ and
$ \stateS \lessThan[\automatonV] \mtop[\automatonV]$;
\item
For all $\statesi \subseteq \states$:
$\glb[\automatonV]  \statesi$ is a lower bound of $\statesi$, and if $\stateSi$ is a lower bound of 
$\statesi$ then $\stateSi\lessThan[\automatonV] \glb[\automatonV]  \statesi$.
%\[
%\glb \statesi \text{ is a lower bound of } \statesi \qquad \land \qquad
%\stateSi \text{ is a lower bound of } \statesi \implies 
%\glb \statesi \lessThan[\automatonV] \stateSi
%\]  
\item
For all $\statesi \subseteq \states$:
$\lub[\automatonV]  \statesi$ is an upper bound of $\statesi$, and if $\stateSi$ is an upper bound of 
$\statesi$ then $\glb[\automatonV]  \statesi \lessThan[\automatonV] \stateSi$.
%\[
%\forall \stateS\in\statesi: \stateS \lessThan[\automatonV] \lub \statesi \qquad \land \qquad
%\stateSi \text{ is a lower bound of } \statesi \implies 
%\lub \statesi \lessThan[\automatonV] \stateSi
%\]  
%\item
%$\stateS[1]\mplus[\automatonV](\stateS[2] \mplus[\automatonV]\stateS[3])
%\mequiv[\automatonV](\stateS[1]\mplus[\automatonV]\stateS[2]) \mplus[\automatonV]
%\stateS[3]$;
%\item
%For all $\statesi \subseteq \states,\stateS$:
%\[\stateS \mplus[\automatonV] \lub[\automatonV] 
%\statesi \mequiv[\automatonV] 
%\lub[\automatonV] \setcomp{\stateS \mplus[\automatonV]\stateSi}{\stateSi \in \statesi}\]
%and  \[\lub[\automatonV] \statesi \mplus[\automatonV] \stateS 
%\mequiv[\automatonV] 
%\lub[\automatonV] \setcomp{\stateSi\mplus[\automatonV]\stateS }{\stateSi \in \statesi}\]
\item
For all $\stateS \in \states, \statesi \subseteq \states: \stateS \mplus[\automatonV] \glb[\automatonV] \statesi 
\mequiv[\automatonV] \glb[\automatonV] \setcomp{\stateS \mplus \stateSi}{\stateSi \in \statesi}$.
\item
For all $\stateS \in \states: \stateS \mplus[\automatonV] \mbot \mequiv[\automatonV] \stateS$.
\item
For all $\stateS,\stateSi \in \states: \stateS \mplus[\automatonV] \stateSi \mequiv[\automatonV] 
\stateSi \mplus \stateS$.
\item
For all $\stateS,\stateSi,\stateSii \in \states: (\stateS \mplus[\automatonV] \stateSi) \mplus[\automatonV] \stateSii \mequiv[\automatonV] 
\stateSi \mplus[\automatonV] (\stateS \mplus[\automatonV] \stateSii)$.
\item
If $\lessThan[\automatonV]$ is a partial order relation, then 
$\automatonV$ is a quantale.
\end{enumerate}
\end{proposition}
%BEGIN LONG VERSION
\longversion{
\begin{proof}
We only show item 4, which is the more involved. So, define 
$\relR \subseteq \states \times \states$ as follows:
\[
\relR = \setcomp{(\stateS,\lub[\automatonV] \statesi)}{\statesi \subseteq \states, \stateS \in \statesi}
\]
We now show that $\relR$ is a $\lessThan[\quantaleQ]$-preserving simulation. So,
let $\stateS\; \relR \lub[\automatonV] \statesi$. Condition 
$\val \stateS \lessThan[\quantaleQ] \val \lub \statesi$
follows from the fact that $\stateS \in \statesi$. For \cref{def:lessThan:item2}, 
suppose $\lub[\automatonV] \statesi \trans[\labelL] \stateS[1]$. By definition of $\lub[\automatonV]$, 
we have that
$\stateS[1] = \lub[\automatonV] \statesii$ for some $\statesii \subseteq \states$ such that 
for all $\stateSi \in \statesi$ there is $\stateSii \in \statesii$
such that $\stateSi \trans[\labelL] \stateSii$. Therefore $\stateS \trans[\labelL] \stateSi$
for some $\stateSi \in \statesii$, and hence $\stateSi \relR \lub[\automatonV] \statesii$. Since
$\lessThan[\automatonV]$ is the largest $\lessThan[\quantaleQ]$-preserving simulation,
we can conclude that $\lub[\automatonV]  \statesi$ is an upper bound of $\statesi$
for all $\statesi \subseteq \states$, as required. 

It remains to show that, for all $\statesi \subseteq \states: \lub[\automatonV] \statesi$ is minimal
among the upper bounds of $\statesi$. So, define 
$\relR \subseteq \states \times \states$ as follows:
\[
\relR = \setcomp{(\lub[\automatonV] \statesi,\stateS)}{\statesi \subseteq \states, \stateS\; 
\text{upper bound of}\; \statesi}
\]
We wish to prove that $\relR$ is a $\lessThan[\quantaleQ]$-preserving simulation.
So, let $\statesi \subseteq \states$ and let $\stateS$ be an upper bound of $\statesi$.
Condition $\val \lub[\automatonV] \statesi \lessThan[\quantaleQ] \val \stateS$ holds by definition.
For \cref{def:lessThan:item2}, suppose $\stateS \trans[\labelL] \stateSi$.
Since $\stateS$ is an upper bound, we have that for all $\stateSii \in \statesi$ there is
$\stateSiii \lessThan[\automatonV] \stateSi$ such that $\stateSii \trans[\labelL] \stateSiii$. Then
define $f$ so that it assigns one such $\stateSiii$ to each $\stateSii \in \statesi$,
and let $\statesii$ be the image of $f$. We have that $\lub[\automatonV] \statesi \trans[\labelL]
\lub[\automatonV] \statesii$. Since $\stateSi$ is an upper bound of $\statesii$, we have that
$\lub[\automatonV] \statesii\; \relR\; \stateSi$, as required.
\end{proof}
}% END LONG VERSION

Unless stated otherwise, we assume that every MLTS $\automatonV$
we work with is a quantale. 
%We recall that every such $\automatonV$ can be 
%turned into a quantale by interpreting $\states$ up-to $\mequiv[\automatonV]$.

\begin{definition}[Contextual Behavioural Metrics]\label{def:cbm}
Let $(\proc,\labels,\trans,\immMetric)$ and 
$\automatonV = (\states,\labels,\trans,\val)$ be, respectively, a 
$\quantaleQ$-LTS and a 
$\quantaleQ$-MLTS. Then, 
a map $\cbmap: \proc \times \proc \to \states$ is a \emph{contextual 
bisimulation map} if:
\begin{enumerate}
\item 
$\immMetric(\procP,\procQ)\;\lessThan[\quantaleQ]\;\;\val\cbmap(\procP,\procQ)$;\label{def:cbm:item1}
\item\label{def:cbm:item2}
if $\cbmap(\procP,\procQ) \trans[\labelL] \stateSi$, then the following holds:
\begin{enumerate}
\item\label{def:cbm:item2:1}
$\procP \trans[\labelL] \procPi \implies \exists \procQi: \procQ \trans[\labelL] \procQi
\;\text{and}\; \cbmap(\procPi,\procQi) \lessThan[\automatonV] \stateSi$;
\item\label{def:cbm:item2:2}
$\procQ \trans[\labelL] \procQi \implies \exists \procPi: \procP \trans[\labelL] \procPi
\;\text{and}\; \cbmap(\procPi,\procQi) \lessThan[\automatonV] \stateSi$.
\end{enumerate}
\end{enumerate}
We say that $\cbmap$ is a \emph{contextual bisimulation metric} (CBM) if $\cbmap$ is 
both a contextual
bisimulation map and a metric. We define the contextual bisimilarity map 
$\metric$ as follows: 
\[
\metric(\procP,\procQ) = \glb[\automatonV]\setcomp{\cbmap(\procP,\procQ)}
{\cbmap\;\text{is a contextual bisimulation map}}
\]
\end{definition}
The following result states that the contextual bisimilarity map is well 
behaved, being a contextual bisimulation map upper bounding any other such map:
\begin{lemma}
$\metric$ is a contextual bisimulation map. Moreover, for all contextual 
bisimulation maps $\cbmap$, and processes $\procP,\procQ$, it holds that
$\metric(\procP,\procQ) \lessThan[\automatonV] \cbmap(\procP,\procQ)$. 
\end{lemma}
%BEGIN LONG VERSION
\longversion{
\begin{proof}
We start by showing that $\metric$ is a contextual bisimulation map. For condition 
$\immMetric(\procP,\procQ) \lessThan[\quantaleQ] \val \metric(\procP,\procQ)$,
notice that $\val \metric(\procP,\procQ) = 
\val \glb[\automatonV] \setcomp{\cbmap(\procP,\procQ)}{\cbmap \text{ is a contextual bisimulation 
map}} = \glb[\quantaleQ] \setcomp{\val \cbmap(\procP,\procQ)}{\cbmap \text{ is a 
contextual bisimulation map}}$. Since $\immMetric(\procP,\procQ) 
\lessThan[\quantaleQ] \val \cbmap(\procP,\procQ)$ for all $\cbmap$, we have the thesis.
For the remaining condition,
suppose $\metric(\procP,\procQ) \trans[\labelL] \stateSi$. By definition of metric and
of $\glb[\automatonV]$, we have that $\cbmap(\procP,\procQ) \trans[\labelL] \stateSi$ for some
contextual bisimulation map $\cbmap$, from which follows the thesis. 

Minimality follows directly from the definition of $\metric$ and of $\glb[\automatonV]$.
\end{proof}
}%END LONG VERSION
We still do not know whether $\metric$ is a \emph{metric}. We need a  
handy characterization of $\metric$ for that.

\paragraph*{A Useful Characterization of CBMs.}
%We now introduce an alternative way of capturing $\metric$, based on a
%quantitative generalization of environment parameterised bisimulations 
%\cite{Larsen87}. It should be noticed that the latter does not involve the 
%immediate metric, and that environments there are just states from an LTS. We 
%instead require parameters to be states of a MLTS this way enabling a metric 
%interpretation. 
Larsen's environment parametrized bisimulations~\cite{Larsen87} is a variation 
on ordinary bisimulation in which the compared states are tested against 
environments of a specific kind, this way giving rise to a ternary relation. We 
here show that CBMs can be captured along the same lines. A formal comparison 
between CBMs and Larsen's approach is 
deferred to \cref{sec:ex-param-bisim}.
\begin{definition}[Parametrized 
Bisimulation]\label{def:parameterised-bisimulation}
%Let $(\proc,\labels,\trans)$ and $(\environments,\labels,\trans')$ be labelled 
%transition systems. The former represents processes, the latter environments.
Let $(\proc,\labels,\trans,\immMetric)$ and 
$(\states,\labels,\trans,\val)$ be, respectively, a 
$\quantaleQ$-LTS and a $\quantaleQ$-MLTS. An $\states$-indexed family of 
relations $\{\relR_{\stateS}\}$ such that $\relR_{\stateS} \subseteq 
\proc \times \proc$ is said to be a \emph{parametrized bisimulation} iff, 
whenever 
$\procP \;\relR_{\stateS}\; \procQ$, it holds that $\immMetric(\procP,\procQ)
\lessThan[\quantaleQ] \val \stateS$, and $\stateS \trans[\labelL] \stateSi$ implies:
\begin{itemize}
\item\label{def:param-bisim-1}
$\procP \trans[\labelL] \procPi \implies \exists \procQi: \procQ \trans[\labelL] \procQi 
\text{ and } \procPi \relR_{\stateSi} \procQi$;
\item
$\procQ \trans[\labelL] \procQi \implies \exists \procPi: \procP \trans[\labelL] \procPi 
\text{ and } \procPi \relR_{\stateSi} \procQi$.\label{def:param-bisim-2}
\end{itemize}
Parametrized bisimilarity is the largest parametrized bisimulation, namely 
the largest family $\{\sim_{\stateS}\}$ such that $\procP \sim_{\stateS} \procQ$ if
$\procP \relR_{\stateS} \procQ$ for some parametrized bisimulation 
$\{\relR_{\stateS}\}$. 
\end{definition}

The fact that $\{\sim_{\stateS}\}$ is indeed a parametrized bisimulation holds 
because parametrized
bisimulations are closed under unions (defined point-wise), something which can 
be proved with a simple
generalisation of standard techniques \cite{Milner89,Park81}. \shortversion{Parametrized bisimilarity turns out to be strongly related to $\metric$, this 
way providing a simple proof technique that will be heavily used in the rest of 
the paper.}
%BEGIN LONG VERSION
\longversion{
The following lemma provides a monotonicity property for parametrized 
bisimilarity, which will be very useful in the following:
\begin{lemma}
If $\stateS \lessThan[\automatonV] \stateSi$ and $\procP \sim_{\stateS} 
\procQ$, then 
$\procP \sim_{\stateSi} \procQ$.
\end{lemma}

\begin{proof}
It suffice to prove that $\relR[\stateS] = \setcomp{(\procP,\procQ)}{\exists \stateSi 
\lessThan[\automatonV] \stateS: \procP \sim_{\stateSi} \procQ}$ is a 
parametrized bisimulation,
which follows easily by the definition of $\lessThan[\automatonV]$. 
\end{proof}

Parametrized bisimilarity turns out to be strongly related to $\metric$, this 
way providing a simple proof technique that will be heavily used in the rest of 
the paper.}%END LONG VERSION
\begin{proposition}\label{lem:metric-bisim}
For all $\procP,\procQ,\stateS$, it holds that
$\metric(\procP,\procQ) \lessThan[\automatonV] \stateS \iff \procP 
\sim_{\stateS} \procQ$.
\end{proposition}
%BEGIN LONG VERSION
\longversion{
\begin{proof}
For the $\Rightarrow$ direction, define the $\states$-indexed family of relations
$\relR_{\stateS} \subseteq \proc \times \proc$ as follows:
\[
\procP\;\relR_{\stateS}\; \procQ \iff \metric(\procP,\procQ) \lessThan[\automatonV] \stateS
\]
It suffice to show that $\relR$ is a parametrized bisimulation. So, let
$\procP \relR_{\stateS} \procQ$. Condition $\immMetric(\procP,\procQ) 
\lessThan[\quantaleQ] \val \stateS$ follows from the fact that 
$\metric(\procP,\procQ) \lessThan[\automatonV] \stateS$.
Now, suppose that $\stateS \trans[\labelL] \stateSi$
and $\procP \trans[\labelL] \procPi$. Since $\stateS \trans[\labelL] \stateSi$,
we have that $\metric(\procP,\procQ) \trans[\labelL] \stateSi[\metric]$.
Then $\procQ \trans[\labelL] \procQi$ for some $\procQi$ such that
$\metric(\procPi,\procQi) \lessThan[\automatonV] \stateSi[\metric]$, that is 
$\procPi \relR_{\stateSi[\metric]} \procQi$. The case for a $\procQ$ move is similar.

For the $\Leftarrow$ direction, it suffice to show that $\cbmap$, defined below,
is a contextual bisimulation map.
\[
\cbmap(\procP,\procQ) = \glb[\automatonV] \setcomp{\stateS}{\procP \sim_{\stateS} \procQ}
\]
Condition $\immMetric(\procP,\procQ) \lessThan[\quantaleQ]
\val \cbmap(\procP,\procQ)$ follows from the fact that $\immMetric(\procP,\procQ)$
is less than  $\val \stateS$ for all $\stateS$ such that $\procP\; \sim_{\stateS}\; \procQ$.
So, suppose $\cbmap(\procP,\procQ) \trans[\labelL] \stateSi$ and 
$\procP \trans[\labelL] \procPi$. Then $\stateS \trans[\labelL] \stateSi$ for some
$\stateS$ such that $\procP\; \sim_{\stateS}\; \procQ$. Hence $\procQ \trans[\labelL]
\procQi$ for some $\procQi$ such that $\procPi\; \sim_{\stateSi}\; \procQi$. Therefore
$\cbmap(\procPi,\procQi) \lessThan[\automatonV] \stateSi$. The case for a $\procQ$ move is similar.
\end{proof}
}%END LONG VERSION
We are finally ready to state that $\metric$ satisfies the axioms of a metric.
\begin{theorem}
The contextual bisimulation map $\metric$ is a metric.
\end{theorem}
%BEGIN LONG VERSION
\longversion{
\begin{proof}
We only show triangle inequality, that is $\metric(\procP[1],\procP[3]) \lessThan[\automatonV]
\metric(\procP[1],\procP[2]) \mplus[\automatonV] \metric(\procP[2],\procP[3])$ for all 
$\procP[1],\procP[2],\procP[3]$. So, define $\relR_{\stateS}$ as follows:
\[
\procP[1] \relR_{\stateS} \procP[3] \iff \exists \procP[2]: 
\stateS = \metric(\procP[1],\procP[2]) \mplus[\automatonV] \metric(\procP[2],\procP[3])
\]
By \cref{lem:metric-bisim} it suffice to prove that 
$\relR_{\stateS}$ is a parametrized bisimulation. So, suppose $\procP[1]
\relR_{\stateS} \procP[3]$. Condition $\immMetric(\procP[1],\procP[3]) 
\lessThan[\quantaleQ] \val (\metric(\procP[1],\procP[2]) \mplus \metric(\procP[2],
\procP[3]))$ follows from the fact that $\immMetric$ is a metric and hence satisfies
triangle inequality. For the remaining condition, suppose $\stateS \trans[\labelL] \stateSi$
and $\procP \trans[\labelL] \procPi$. Then $\stateSi = \stateSi[1] \mplus[\automatonV] \stateSi[2]$
for some $\stateSi[1],\stateSi[2]$ such that $\stateS[1] \trans[\labelL] \stateSi[1]$ and
$\stateS[2] \trans[\labelL] \stateSi[2]$. Thus $\procP[2] \trans[\labelL] \procPi[2]$
for some $\procPi[2]$ such that $\metric(\procPi[1],\procPi[2]) \lessThan[\automatonV] \stateSi[1]$.
This in turn implies $\procP[3] \trans[\labelL] \procPi[3]$
for some $\procPi[3]$ such that $\metric(\procPi[2],\procPi[3]) \lessThan[\automatonV] \stateSi[2]$.
Therefore $\procPi[1] \relR[\stateSi] \procPi[3]$, as required. The case for $\procP[3]$
moves is similar.
\end{proof}
}%END LONG VERSION

\section{Some Relevant Examples}\label{sect:Examples}
%!TeX root=main.tex
This section is devoted to showing how well-known and heterogeneous notions of 
equivalence and distance can be recovered as CBMs for appropriate quantales and 
MLTSs.

\subsection{Strong Bisimilarity as a CBM}\label{ex:bisimilarity}

We start recalling that strong bisimilarity \cite{Milner89,Park81} is the 
largest strong bisimulation relation,
that is a relation $\relR \subseteq \proc \times \proc$ on the states of a
plain LTS $(\proc,\labels,\trans)$ such that 
$\procP\;\relR\;\procQ$ implies:
\begin{itemize}
\item
$\procP \trans[\labelL] \procPi \implies \exists \procQi: \procQ \trans[\labelL] \procQi\;
\text{and}\;\procPi\;\relR\;\procQi$;
\item
$\procQ \trans[\labelL] \procQi \implies \exists \procPi: \procP \trans[\labelL] \procPi\;
\text{and}\;\procPi\;\relR\;\procQi$.
\end{itemize}
The first thing we have to do to turn strong bisimilarity into a CBM is to 
define, given such an LTS $(\proc,\labels,\trans)$, a canonical immediate 
distance $\immMetric$ on the boolean quantale $\quantaleB$, which we call the 
\emph{canonical} distance:
\[\immMetric(\procP,\procQ) = 
\begin{cases}
\mbot \text{ if } \forall \labelL: \procP \trans[\labelL] \iff  \procQ \trans[\labelL] \\
\mtop \text{ otherwise}
\end{cases}\]
That is, the immediate distance is $\mbot$ precisely when the processes expose 
the same labels. Notice that immediate distance is not affected
by possible future behavioural differences. Any LTS like this is said to be a 
\emph{boolean} LTS. The boolean quantale can be turned very naturally into a 
MLTS: let $\automatonV$  be 
$(\{\mbot[\automatonV],\mtop[\automatonV]\},\labels,\trans,\val)$ where
the transitions are self loops $\mbot[\automatonV] \trans[\labelL] \mbot[\automatonV]$ for every $\labelL \in \labels$,
and $\val$ just associates $\mbot[\quantaleQ]$ to $\mbot[\automatonV]$ and 
$\mtop[\quantaleQ]$ to $\mtop[\automatonV]$. 
%\UDL{Meglio spiegare cosa fa $\trans$ e cosa fa $\val$, brevemente.}
%where:
%\[
%\transF(\mbot[\automatonV],\labelL) = \mbot[\automatonV] \qquad
%\transF(\mtop[\automatonV],\labelL) = \mtop[\automatonV] \qquad
%\val \mbot[\automatonV] = \mbot \qquad \val \mtop[\automatonV] = \mtop
%\]
%Notice that $\automatonV$ has only two elements and it is basically the 
%boolean quantale, although expressed as a Metric LTS.\\
\begin{proposition}
Given any boolean LTS, $\metric$ is the characteristic function of
bisimilarity, i.e. 
$
\metric(\procP,\procQ) = \mbot[\automatonV] \iff \procP \sim \procQ 
%\quad \text{and} \quad \metric(\procP,\procQ) = \mtop[\automatonV] \iff \procP \not\sim \procQ
$.
\end{proposition}
%\UDL{Lo sketch di prova che segue meglio metterlo in versione lunga.}
%BEGIN LONG VERSION
\longversion{
\begin{proof}
To see why, define: 
\[\cbmap(\procP,\procQ)= 
\begin{cases}
\mbot[\automatonV] \text{ if } \procP \sim \procQ\\
\mtop[\automatonV] \text{ otherwise}
\end{cases}
\]
Notice that $\cbmap$ is a contextual bisimulation map. Indeed, condition 
$\immMetric(\procP,\procQ) \lessThan[\quantaleQ] \val \cbmap(\procP,\procQ)$
surely holds: if $\immMetric(\procP,\procQ) = \mtop[\quantaleQ]$ it must be
$\procP \not\sim \procQ$ and hence $\val \cbmap(\procP,\procQ) = \mtop[\quantaleQ]$.
For the other condition, suppose $\cbmap(\procP,\procQ) \trans[\labelL] \stateSi$. Then 
$\cbmap(\procP,\procQ) = \stateSi = \mbot[\automatonV]$. Therefore $\procP \sim \procQ$,
and hence if $\procP \trans[\labelL] \procPi$ there is a matching $\labelL$-transition
of $\procQ$, and viceversa.
From the above, together with the fact that $\metric$ is minimal, it follows that: 
\[
\procP \sim \procQ \implies \metric(\procP,\procQ) = \mbot[\automatonV]
\]
%\begin{itemize}
%\item $\procP \sim \procQ \implies \metric(\procP,\procQ) = \mbot$;
%\item $\metric(\procP,\procQ) = \mtop \implies \procP \not\sim \procQ$.
%\end{itemize}
For the reverse implication, define $\relR \subseteq \proc \times \proc$ as follows:
\[
\procP\;\relR\;\procQ \iff \metric(\procP,\procQ) = \mbot[\automatonV]
\]
It suffice to show that $\relR$ is a bisimulation relation.
So, suppose $\procP\;\relR\;\procQ$ and $\procP \trans[\labelL] \procPi$. Since
$\metric(\procP,\procQ) = \mbot[\automatonV]$, we have that 
$\metric(\procP,\procQ) \trans[\labelL] \mbot[\automatonV]$ and hence 
$\procQ \trans[\labelL] \procQi$ for some $\procQi$ such that
$\metric(\procPi,\procQi) = \mbot[\automatonV]$. 
Therefore $\procPi \;\relR\;\procQi$. 
The case for a $\procQ$ move is similar.
\end{proof}
} %END LONG VERSION

\subsection{Behavioural CBMs}\label{ex:beh-bisimilarity}
Most behavioural metrics from the literature are defined on
\emph{probabilistic} 
transition systems~\cite{DesharnaisGJP04,DesharnaisLT08,BreugelW06}, 
differently from CBMs. Some probabilistic behavioural metrics can still be 
captured in our framework by using as states of the process LTS
(sub)distributions of states of the original PLTS, e.g. the distribution 
based metric in \cite{DuDG22}.
Non-probabilistic behavioural metrics exist, e.g., the 
so-called ``branching metrics''~\cite{AlfaroFS04}, which 
are indeed instances of behavioural metrics as defined below. Notice that our 
definition has a generic quantale 
$\quantaleQ$ as its codomain, while usually behavioural metrics take values in 
the interval $\mathbb{R}_{[0,1]}$.

Let us first recall what we mean by a behavioural metric here. A metric $M: 
\proc \times \proc 
\to \quantaleQ$ is said to be a \emph{behavioural metric} if, for all pairs of 
states $\procP,\procQ$, it holds that 
$\immMetric(\procP,\procQ) \lessThan[\quantaleQ] M(\procP,\procQ)$ and, whenever 
$M(\procP,\procQ) \lessThanStrict[\quantaleQ] \mtop[\quantaleQ]$, we have that:
\begin{itemize}
\item
$\procP \trans[\labelL] \procPi \implies \exists \procQi: \procQ \trans[\labelL] \procQi\;
\text{and}\;M(\procP,\procQ) \greatThan[\quantaleQ] M(\procPi,\procQi)$;
\item
$\procQ \trans[\labelL] \procQi \implies \exists \procPi: \procP \trans[\labelL] \procPi\;
\text{and}\;M(\procP,\procQ) \greatThan[\quantaleQ] M(\procPi,\procQi)$.
\end{itemize}

Intuitively, behavioural metrics can be seen as quantitative variations on the 
theme of a bisimulation: they
associate a value from a quantale to each pair of processes (rather than a 
boolean), they are coinductive in nature. Moreover, they are based on the 
bisimulation game, i.e., any move of one of the two processes needs to be 
matched by some move of the other, at least when their distance is not maximal.
\highlight{Our definition is similar to the one in \cite{DuDG22}.}
However, many behavioural metrics in literature deal with non-determinism through
the Hausdorff lifting, that is by stipulating that $\immMetric(\procP,\procQ) 
\lessThan[\quantaleQ] M(\procP,\procQ)$ and
for all $\labelL$:
\[
M(\procP,\procQ) \greatThan[\quantaleQ] \lub[{\procP \trans[\labelL] \procPi}] \glb [{\procQ \trans[\labelL] \procQi}] M(\procPi,\procQi)\qquad
\text{and}\qquad M(\procP,\procQ) \greatThan[\quantaleQ]\lub[{\procQ \trans[\labelL] \procQi}] \glb [{\procP \trans[\labelL] \procPi}] M(\procPi,\procQi)
\]
The two notions are equivalent if the process LTS is image-finite and $\quantaleQ$ is totally ordered, both conditions are often assumed to be true in the literature.
%BEGIN LONG VERSION
\longversion{
The following lemma states the above formally.
\begin{lemma}
Over totally ordered quantales and for image-finite process LTSs, behavioural 
metrics and Hausdorff metrics coincide.
\end{lemma}
\begin{proof}
We first show that behavioural metrics (BM for short) are Hausdorff metrics (HM 
for short).
So, let $M$ be a BM. The requirement $\immMetric(\procP,\procQ) \lessThan[\quantaleQ] M(\procP,\procQ)$ holds by definition of BM. 
Let's focus to the requirement $M(\procP,\procQ) \greatThan[\quantaleQ] \lub[{\procP \trans[\labelL] \procPi}] \glb [{\procQ \trans[\labelL] \procQi}] M(\procPi,\procQi)$.
First notice that it holds trivially when $M(\procP,\procQ) = \mtop$. Otherwise, it suffice to show that $M(p,q) \greatThan[\quantaleQ] 
\glb [{\procQ \trans[\labelL] \procQi}] M(\procPi,\procQi)$ whenever $\procP \trans[\labelL] \procPi$. So, suppose $\procP \trans[\labelL] \procPi$.
Then, by definition of BM, there is $\procQi$ such that $\procQ \trans[\labelL] \procQi$ and $M(\procP,\procQ) \greatThan[\quantaleQ] M(\procPi,\procQi) = d$. The thesis
follows as $d \greatThan[\quantaleQ] \glb [{\procQ \trans[\labelL] \procQi}] M(\procPi,\procQi)$ by definition of $\glb$. The requirement
$M(\procP,\procQ) \greatThan[\quantaleQ] \lub[{\procQ \trans[\labelL] \procQi}] \glb [{\procP \trans[\labelL] \procPi}] M(\procPi,\procQi)$ follows by
a similar argument. We now show that HMs are BMs, provided the process LTS is image-finite and $\quantaleQ$ is a total order.
So, let $M$ be a HM. The requirement $\immMetric(\procP,\procQ) \lessThan[\quantaleQ] M(\procP,\procQ)$ holds by definition of HM. We consider the case 
$M(\procP,\procQ) \lessThanStrict[\quantaleQ] \mtop[\quantaleQ]$ as otherwise the thesis is trivial. So, let $\procP \trans[\labelL] \procPi$. By definition of
HM, we have that $M(\procP,\procQ) \greatThan[\quantaleQ] \glb [{\procQ \trans[\labelL] \procQi}] M(\procPi,\procQi)$. 
Since $S = \setcomp{M(\procPi,\procQi)}{\procQ \trans[\labelL] \procQi}$ is finite (by image-finiteness) and totally ordered, we have that 
$\glb [{\procQ \trans[\labelL] \procQi}] M(\procPi,\procQi) \in S$, from which the thesis follows.
\end{proof}
}%END LONG VERSION
%We then define $\bar{m}: \proc \times \proc \to \quantaleQ$ as follows:
%\[
%\bar{m}(\procP,\procQ) = \glb \setcomp{m(\procP,\procQ)}{m\;\text{is a 
%behavioural metric}}
%\]

We now show how to interpret $\quantaleQ$ as a MLTS. Morally, we just fix 
$\quantaleQ$ as the
set of states, the identity as $\val$, and self loops as transitions. This 
however violates
the requirement that the top element has no outgoing transitions. 
We therefore add the element
$\mtop[\automatonV]$. Notice that we still need $\mtop[\quantaleQ]$, as it ensures 
that $\automatonV$
is closed under $\mplus[\automatonV]$.
%Define $\automatonV = (\states,\labels,\trans,\val)$ as follows:
%\begin{itemize}
%\item
%$\states = \quantaleQ \cup \{\mtop[\automatonV]\}$, where 
%$\mtop[\automatonV] \not\in \quantaleQ$;
%\item
%$\labels$ is as in the process LTS;
%\item
%$\trans = \setcomp{(\stateS \trans[\labelL] \stateS)}{\stateS \in \quantaleQ,\labelL 
%\in \labels}$;
%\item
%$\val \stateS = \begin{cases} \stateS & \text{if}\;\stateS\in\quantaleQ \\
%\mtop[\quantaleQ] & \text{otherwise}           
%\end{cases}$
%\end{itemize}
Let $\automatonV = (\states,\labels,\trans,\val)$ where 
$\states = \quantaleQ \uplus \{\mtop[\automatonV]\}$, $\labels$ is as in the 
underlying process LTS,
transitions are the self loops of the form $\stateS \trans[\labelL] \stateS$ 
for every
$\labelL \in \labels$, and $\stateS \in \quantaleQ$, $\val$ is the identity on 
$\quantaleQ$, and $\val(\mtop[\automatonV]) = \mtop[\quantaleQ]$.
Notice that, when $\val \stateS \lessThanStrict[\quantaleQ] \mtop[\quantaleQ]$, we have that:
\begin{equation}\label{ex:beh-bisimilarity:eq1}
\val \stateS \lessThan[\quantaleQ] \val \stateSi  \iff \stateS \lessThan[\automatonV] \stateSi.
\end{equation}
%\UDL{Cercherei di far diventare il risultato sotto una proposizione, come nel 
%caso booleano, e metterei la dimostrazione in versione lunga, magari 
%rimpiazzandola con una versione più breve.}
We also have that for every behavioural metric there is a CBM that ``agrees'' 
on the quantitative distance between processes. This intuition is 
formalized as follows:

\begin{proposition}
Let $M$ be a behavioural metric, and let $m_{M}$ be defined as:
\[m_{M}(\procP,\procQ) = \begin{cases} M(\procP,\procQ) & \text{if}\;  
M(\procP,\procQ) \lessThanStrict[\quantaleQ] \mtop[\quantaleQ]\\
\mtop[\automatonV] & \text{otherwise}
\end{cases}\]
Then, $m_{M}$ is a CBM and for every $\procP,\procQ$ it holds that
$\val m_{M}(\procP,\procQ) = M(\procP,\procQ)$.
\end{proposition}
%BEGIN LONG VERSION
\longversion{
\begin{proof}
The fact that $\val m_{M}(\procP,\procQ) = M(\procP,\procQ)$ follows immediately from the definition. It remains to show that
$m_{M}$ is really a CBM. We first show that it is a contextual bisimulation map. Notice that $\immMetric(\procP,\procQ)
\lessThan[\quantaleQ] \val m_{M}(\procP,\procQ)$ follows immediately from the 
definitions of $m_{M}$ and behavioural metrics.
So, suppose $m_{M}(\procP,\procQ) \trans[\labelL] \stateSi$ and $\procP \trans[\labelL] \procPi$. Then, 
$\val m_{M}(\procP,\procQ) = M(\procP,\procQ) \lessThanStrict[\quantaleQ] \mtop[\quantaleQ]$ and 
$m_{M}(\procP,\procQ) = \stateSi$. Therefore $\procQ \trans[\labelL] \procQi$ for some $\procQi$ such that 
$M(\procPi,\procQi) \lessThan[\quantaleQ] M(\procP,\procQ)$. Since 
$\val m_M(\procPi,\procQi) = M(\procPi,\procQi) \lessThanStrict[\quantaleQ] \mtop[\quantaleQ]$, we have that
$m_M(\procPi,\procQi) \lessThan[\quantaleQ] m_{M}(\procP,\procQ) = \stateSi$, as required. The case for a $\procQ$ move is similar.
It remains to show that $m_M$ is a metric. We only show triangle inequality, that is:
$m_M(\procP[1],\procP[3]) \lessThan[\automatonV] m_M(\procP[1],\procP[2]) \mplus[\automatonV] m_M(\procP[2],\procP[3])$.
If $\val m_M(\procP[1],\procP[3]) \lessThanStrict[\quantaleQ] \mtop[\quantaleQ]$, the thesis follows from the fact that $M$ is a metric
and \cref{ex:beh-bisimilarity:eq1}.
If $\val m_M(\procP[1],\procP[3]) = \mtop[\quantaleQ]$, it must be $m_M(\procP[1],\procP[3]) = \mtop[\automatonV]$. 
Since $M$ is a metric, we have that $\val m_M(\procP[1],\procP[2]) = \mtop[\quantaleQ]$ or 
$\val m_M(\procP[2],\procP[3]) = \mtop[\quantaleQ]$. Then  $m_M(\procP[1],\procP[3]) = \mtop[\automatonV]$ or
$m_M(\procP[2],\procP[3]) = \mtop[\automatonV]$. In both cases $m_M(\procP[1],\procP[2]) \mplus[\automatonV] 
m_M(\procP[2],\procP[3]) = \mtop[\automatonV]$, as required.
\end{proof}
}%END LONG VERSION
\shortversion{%BEGIN SHORT VERSION
The agreement of $M$ and $m_{M}$ holds by definition. The fact that $m_{M}$ is 
a CBM, instead is a consequence of the fact that transitions preserve $M$ 
distances (by definition of behavioural metric), and that behavioural metrics 
are metrics, indeed.
}%END SHORT VERSION

\subsection{On Environment Parametrised Bisimulation and CBMs}
\label{sec:ex-param-bisim}
As already mentioned, the concept of a CBM is inspired by Larsen's 
environment parametrized bisimulation~\cite{Larsen87}. It should then come 
with no surprise that there is a relationship between the two, which is the 
topic of this section.

First, let us recall what an environment parametrized bisimulation is. Let 
$(\proc,\labels,\trans)$ and $(\env,\labels,\trans)$ be LTSs. Elements of 
$\proc$ are called processes, while elements of $\env$ are called environments.
A $\env$-indexed family of relations $\{\relR_{\elE}\}$, where $\relR_{\elE} 
\subseteq 
\proc \times \proc$ is a \emph{environment parametrized bisimulation} (EPB in 
the following) if, whenever 
$\procP \;\relR_{\elE}\; \procQ$ and 
$\elE \trans[\labelL] \elEi$:
\begin{itemize}
\item
$\procP \trans[\labelL] \procPi \implies \exists \procQi: \procQ \trans[\labelL] 
\procQi \text{ and } \procPi \relR_{\elEi} \procQi$;
\item
$\procQ \trans[\labelL] \procQi \implies \exists \procPi: \procP \trans[\labelL] 
\procPi \text{ and } \procPi \relR_{\elEi} \procQi$.
\end{itemize}
Environment parametrized bisimilarity, denoted as 
$\sim_{\elE}$, is defined as $\procP \sim_{\elE} \procQ$ iff $\procP\; 
\relR_{\elE}\;\procQ$ for
some EPB $\relR$. It turns out that $\sim_{\elE}$ is the largest
EPB~\cite{Larsen87}.
%\UDL{Di nuovo, qui 
%sopra penso valga la pena di dire come riusciamo a costruire la relazione più 
%grande.}

EPBs can be embedded into the CBMs framework as follows:
\begin{itemize}
\item
fix $\quantaleQ$ as the boolean quantale $\quantaleB$, and define 
$\immMetric(\procP,\procQ) = \begin{cases}
\mbot[\quantaleB] & \text{if}\;\exists \labelL: \procP \trans[\labelL]\;\text{and}\;
\procQ \trans[\labelL]\\
\mtop[\quantaleB] & \text{otherwise} \end{cases}$
\item
let $\automatonV_{\env} = (\states,\labels,\trans,\val)$ be any MLTS such that for all 
$
\stateS \in \states$ it holds that
$\val \stateS = \mbot[\quantaleB] \iff \stateS \neq \mtop[\automatonV]$, and 
for all $
\elE \in \env$
there is $\stateS[\elE] \in \states$ such that $\elE \mutsimilar \stateS[\elE]$. Here 
$\mutsimilar$ is strong mutual similarity on the disjoint union of 
$\automatonV$ (forgetting $\val$) and $\env$. When such conditions hold, we say that $\env$
is embedded into $\automatonV_{\env}$.
\end{itemize}
We remark that, for every $\env$, 
there is an MLTS $\automatonV_{\env}$ enjoying the properties above, obtained by
augmenting $\env$ with the immediate metric defined above (this gives rise to a pre-metric LTS,
\cref{def:PMLTS}) 
and by closing it with respect to the operations and constants $\lub,\glb,\mtop,\mbot$ of
\cref{def:MLTS}. 
%\UDL{Qui sopra non è molto chiaro come sia costruito $\mathbb{V}$ e in 
%particolare come sia costruito l'insieme degli stati. Tra l'altro non non mi è 
%chiaro cosa sia $\mutsimilar$. Ultima cosa: per ogni $E$ possiamo sempre 
%costruirlo $\mathbb{V}$?}
%\MM{Ho aggiunto un remark, spero risponda alla domanda.}
The intuition is that:
\begin{itemize}
\item Two processes should have minimal immediate distance if there is a non-empty 
context in which their immediate behaviour is equivalent. This is 
ensured
by the fact that they exhibit at least a common label from their current state.
\item $\automatonV_{\env}$ needs to precisely simulate the behaviours 
in $\env$. We therefore require that every element of $\env$ has a 
corresponding 
element in $\stateS$, with ``equivalent behaviour''. In this setting, 
mutual simulation turns out to be the appropriate notion of behavioural 
equivalence. 
\end{itemize} 

The link between environment parametrized bisimulations and CBMs is made formal
by the following proposition.
\begin{proposition}
Let $\env$ be an environment LTS embedded into an MLTS $\automatonV_{\env}$. For every $\procP,\procQ$ and $\elE$, it holds that 
$\procP \sim_{\elE} \procQ \implies \metric(\procP,\procQ) 
\lessThan[\automatonV_{\env}] 
\stateS[\elE]
$.
\end{proposition}
The proposition above ultimately follows from the fact that $\procP \sim_{\elE} 
\procQ \iff \procP \sim_{\stateS[\elE]} \procQ$ (where $\sim_{\stateS[\elE]}$ 
is parametrized bisimilarity 
\cref{def:parameterised-bisimulation}) together with \cref{lem:metric-bisim}. 
%\UDL{L'ultimo statement, al solito, lo farei diventare una proosizione.}

\section{About the Compositionality of CBMs}\label{sect:Composition}
%!TeX root=main.tex
One of the greatest advantages of the bisimulation proof method is its modularity, 
which comes from the fact that, under reasonable assumptions, bisimilarity is a 
congruence. In a metric setting, one strives to obtain similar 
properties~\cite{GeblerLT16,GeblerT18}, which take the form of 
non-expansiveness, or variations thereof.

In this section we study the compositionality properties of CBMs with respect to some standard 
process algebraic operators. We are interested in properties that generalise 
the concept of a congruence. Following the lines of
\cite{LagoGY19,Pistone21}, our treatment will be contextual,
meaning that the environment in which processes are deployed can indeed contribute 
to altering their distance, although in a controlled way.

In order to keep our theory syntax independent, we model operators $f$ as functions
$f:\proc^n \to \proc$ (where $n$ is the arity of the operator).
In particular, for each process operator $f$ of arity $n$ we define the function
$\hat{f}: \proc^{n} \times \states^n \to \states$ as follows:
\[
\hat{f} (\procP[1],\hdots,\procP[n],\stateS[1],\hdots,\stateS[n]) = 
\lub[\automatonV]{\setcomp{\metric(f(\procP[1],\hdots,\procP[n]),
f(\procPi[1],\hdots,\procPi[n]))}{\forall 1 \leq i \leq n: \metric(\procP[i],\procPi[i])
\lessThan[\automatonV] \stateS[i]}}
\]
%\[
%\metric(f(c_1,\hdots,c_{i-1},\procP,c_{i-1},\hdots,c_{n}),
%f(c_1,\hdots,c_{i-1},\procQ,c_{i-1},\hdots,c_{n})) \lessThan
%\hat{f}_i (c_1,\hdots,c_{i-1},c_{i-1},\hdots,c_{n}) \metric(\procP,\procQ)
%\]
Intuitively, $\hat{f}(\vec{\procP},\vec{\stateS})$
bounds $\metric(f(\vec{\procP}),f(\vec{\procQ}))$ whenever $\vec{\procQ}$ is such that 
$\metric(\procP[i],\procQ[i]) \lessThan[\automatonV] \stateS[i]$ for every $i$. 
Moreover, $\hat{f}(\vec{\procP},\vec{\stateS})$ is the lowest among such bounds.

Of course, our compositionality results rely on some assumptions on the
compositionality of the immediate metric $\immMetric$. Formally, we require that,
for all operators $f$ (with arity $n$), the following holds for every
$\procP[1],\hdots,\procP[n],\procQ[1],\hdots,\procQ[n]$:
\begin{equation}\label{eq:comp-imm}
\immMetric(f(\procP[1],\hdots,\procP[n]),f(\procQ[1],\hdots,\procQ[n]))
\lessThan[\quantaleQ] \immMetric(\procP[1],\procQ[1]) \mplus[\quantaleQ] \hdots
\mplus[\quantaleQ] \immMetric(\procP[n],\procQ[n]).
\end{equation}

%Moreover, we assume that $\hat{f}$ satisfies a sort of non-expansiveness propertiy 
%with respect to immediate distance $\immMetric$. Formally, we require that,
%for all $\procP[1],\hdots,\procP[n],\stateS$ and $1 \leq i \leq n$:
%\[
%\val \hat{f}(\procP[1],\hdots,\procP[n],\overbrace{\mbot,\hdots,\mbot}^{i-1},\stateS,
%\overbrace{\mbot,\hdots,\mbot}^{n-i}) \lessThan \val{\stateS}
%\] 

%We consider uniform continuity, defined below, as a minimal condition for an operator to 
%be dubbed as compositional. 
%
%\begin{definition}[Uniform continuity]
%An operator $f:\proc ^ n \to \proc$ is uniformly continuous \wrt a CBM $\metric$ if
%for all $\stateS \greatThan \mbot$ there are 
%$\stateS[1],\hdots,\stateS[n] \greatThan \mbot$ such that:
%\[
%\forall 1\leq i \leq n: \metric(\procP[i],\procQ[i]) \lessThan \stateS[i] \implies
%\metric(f(\procP[1],\hdots,\procP[n]),f(\procQ[1],\hdots,\procQ[n])) \lessThan \stateS
%\]
%\end{definition}

Below, we will give results about when and under which condition the value of the operator $\hat{f}$ can be upper-bounded by a function on its parameters. We remark that
our compositionality results apply to each operator independently.

For the sake of concreteness, we give some examples of processes 
and their metric analysis. To this purpose, let $\labels = 
\{\labelA,\labelB\}$, fix $\quantaleQ$ as the boolean 
quantale $\quantaleB$ and let $\immMetric$ be defined exactly as we did in 
\cref{sec:ex-param-bisim} (i.e., $\immMetric$ returns $\mbot$ if the processes 
can fire some common action, $\mtop$ 
otherwise). Distances will take values from a MLTS 
$\automatonV_0$ over $\quantaleB$. Similarly to \cref{sec:ex-param-bisim}, we require 
$\automatonV_0$ to be such that 
for every $\stateS \in \states$ it holds that $\val \stateS = \mbot[\quantaleB] 
\iff \stateS \neq 
\mtop[\automatonV_0]$. Moreover, we assume that $\automatonV_0$ is able to represent
at least Milner's synchronisation trees \cite{Milner80}. 
%This inuitively guarantees that it can capture
%every finite state and every finitely branching behaviour (up-to mutual 
%similarity).
For simplicity, we omit self loops of $\mbot[\automatonV_0]$ from
all the graphical representations of our MLTS. 
Of course, these assumptions hold only in the examples, while our results hold for
general MLTSs.

%\begin{itemize}
%\item
%$\hat{f}_i(\vec c)\stateS \lessThan \stateS$ (non-expansivity);
%\item
%$\hat{f}_i(\vec c)\bot \lessThan \bot$ (bottom preservation);
%\item
%$\stateS \lessThan \stateSi \implies 
%\hat{f}_i(\vec c)\stateS \lessThan \hat{f}_i(\vec c)\stateSi$ (monotonicity);
%\item
%$\top \lessThan \hat{f}_i(\vec c)\stateS \implies \top \lessThan \stateS$ (TODO: find a proper name). 
%\end{itemize} 
%Non-expansivity is the strongest property, and indeed implies the others.
%Bottom preservation basically means that the kernel of the metric is a congruence.
%Monotonicity is a natural compositionality property as it ensures that if a component
%is substituted by another one in a certain system, than the more the components are 
%"close" the more the two systems are. The last property ensures that the operator
%does not give

\subsection{Restriction}
We assume restriction to be modelled by a $\labels$-indexed family of unary 
operators $\new_{\labelL}$, and that $\proc$ is closed under these operators.
\shortversion{
Their semantics is standard.}
\longversion{
Their semantics can be defined in a standard way:
%$\new: \labels \times \proc \to \proc$, whose semantics is defined below:
\[
\irule{\procP \trans[\labelLi]\procPi \qquad \labelLi \neq \labelL}
{\new_{\labelL}\;\procP \trans[\labelLi]\new_{\labelL}\;\procPi}{}
\]}
%Technically, $\new$ is not of the form $f:\proc^n \to \proc$ due to the 
%presence
%of a label. However, we can model $\new$ as a $\labels$-indexed family of unary
%operators $\new_{\labelL}$.

\begin{example}\label{ex:restriction}
%Let the set of labels $\labels = \{\labelA,\labelB\}$, fix $\quantaleQ$ as the boolean 
%quantale $\quantaleB$ and let $\immMetric$ be defined exactly as we did in 
%\cref{sec:ex-param-bisim} ($\mbot$ if the processes can fire some common action, $\mtop$ 
%otherwise). We do not define 
%$\automatonV$ formally, but we assume it is any MLTS over $\quantaleB$ that contains at least
%every finite state behaviour immaginable (up-to mutual similarity). We will 
%assume such setup
%for all the examples in this section, and we omit self loops of $\mbot[\automatonV]$ from
%the figures. 
Let $\procP[0]$ and $\procQ[0]$ be as in the following 
figure. We have that $\procP[0]$
and $\procQ[0]$ have the exact same behaviour on the $\labelB$ branch, while we 
can observe
differences on the $\labelA$ branch ($\procQ[1]$ can perform an action, $\procP[1]$ is 
terminated). State $\stateS[0]$ captures exactly the similarities between $\procP[0]$
and $\procQ[0]$: after a $\labelB$ move it reduces to $\mbot$; after an $\labelA$ move,
it reduces to $\stateS[1]$. We argue that $\stateS[1]$ captures the similarities between
$\procP[1]$ and $\procQ[1]$: since neither of the two can perform the action $\labelA$,
$\stateS[1]$ reduces to $\mbot$ with label $\labelA$, while it does not perform $\labelB$
actions because $\procP[1]$ and $\procQ[1]$ ``disagree'' on such label. 
So $\metric(\procP[0],\procQ[0])=\stateS[0]$. 
Processes $(\new \labelA) \procP[0]$ and $(\new \labelA) \procQ[0]$ exhibit
equivalent behaviour instead. In fact, operator $(\new \labelA)$ filters out the
problematic $\labelA$ branch. It is therefore the case that 
$\metric((\new \labelA)\procP[0],(\new \labelA)\procQ[0])=\mbot$.
\begin{center}
	\includegraphics{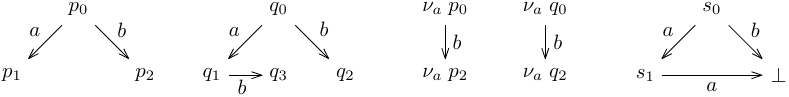}
	\commentout{
\begin{tikzpicture}[node distance = 2cm, on grid,auto]
\node (q0) [state][label =above:{$\procP[0]$}] {};
\node (q1) [state, below left = of q0][label =above:{$\procP[1]$}] {};
\node (q2) [state, below right = of q0][label =above:{$\procP[2]$}] {};
\node (p0) [state,right =4.5cm of q0][label =above:{$\procQ[0]$}] {};
\node (p1) [state, below left = of p0][label =above:{$\procQ[1]$}] {};
\node (p2) [state, below right = of p0][label =above:{$\procQ[2]$}] {};
\node (p3) [state, below = of p0][label =above:{$\procQ[3]$}] {};
\node (s0) [state,right =3.5cm of p0][label =above:{$\stateS[0]$}] {};
\node (s1) [state, below = of s0][label =left:{$\stateS[1]$}] {};
\node (s2) [state, below right = of s0][label =right:{$\mbot$}] {};
\path [->]
(q0) edge [] node {$\labelA$}   (q1)
(q0) edge [] node {$\labelB$}   (q2)
(p0) edge [] node {$\labelA$}   (p1)
(p0) edge [] node {$\labelB$}   (p2)
(p1) edge [] node {$\labelB$}   (p3)
(s0) edge [] node {$\labelA$}   (s1)
(s0) edge [] node {$\labelB$}   (s2)
(s1) edge [] node {$\labelA$}   (s2)
;
\end{tikzpicture}
\\[5pt]
\begin{tikzpicture}[node distance = 2cm, on grid,auto]
\node (q0) [state][label =above:{$(\new \labelA) \procP[0]$}]  {};
\node (q2) [state, right = of q0][label =above:{$(\new \labelA) \procP[2]$}] {};
\node (p0) [state,right =4.5cm of q0][label =above:{$(\new \labelA) 
\procQ[0]$}] {};
\node (p2) [state, right = of p0][label =above:{$(\new \labelA) \procQ[2]$}] {};
\path [->]
(q0) edge [] node {$\labelB$}   (q2)
(p0) edge [] node {$\labelB$}   (p2)
;
\end{tikzpicture}}
\end{center}
%\vspace{-0.8cm}
\qed
\end{example}

The restriction operator does not add new behaviours to the original process,
as it can only restrict it. We can then expect that the differences between any two
processes do not increase if such processes are placed in a restriction context.
Proposition below indeed shows that $\hat{\new_{\labelL}}$ enjoys a property
similar to non-expansiveness, that is the distance between any two processes $\procP$ and $\procQ$ bounds
the distance between $\new_{\labelL}\; \procP$ and $\new_{\labelL}\; \procQ$.

%\UDL{Questa frase dovrebbe essere completata. Tra l'altro, l'ipotesi che 
%facciamo sull'MLTS dovrebbe secondo me essere esplicitata per bene all'inizio 
%della sezione, perché mi sembra importante.}
\begin{proposition}
$\hat{\new_{\labelL}}(\procP,\stateS) \lessThan[\automatonV] \stateS$.
\end{proposition}
%BEGIN LONG VERSION
\longversion{
\begin{proof}
It suffice to show that $\relR_{\stateS} = \setcomp{(\new_{\labelL} \procP,
\new_{\labelL} \procQ)}{\metric(\procP,\procQ) \lessThan[\automatonV] \stateS}$ 
is a parametrized
bisimulation. So, let $\new_{\labelL} \procP\; \relR_{\stateS}\; \new_{\labelL} \procQ$.
Condition $\immMetric(\new_{\labelL} \procP,\new_{\labelL} \procQ) 
\lessThan[\quantaleQ] \val \stateS$ follows by \cref{eq:comp-imm}. For condition 2,
suppose $\stateS \trans[\labelLi] \stateSi$ and $\new_{\labelL} \procP \trans[\labelLi]
\hat{\procP}$. By inversion, we have that $\labelL \neq \labelLi$ and $\hat{\procP} = 
\new_{\labelL} \procPi$ for some $\procPi$ such that $\procP \trans[\labelLi] \procPi$.
Then $\procQ \trans[\labelLi] \procQi$ for some $\procQi$ such that 
$\metric(\procPi,\procQi) \lessThan[\automatonV] \stateSi$. Therefore:
\[
\irule{\procQ \trans[\labelLi]\procQi \qquad \labelLi \neq \labelL}
{\new\; \labelL\; \procQ \trans[\labelLi]\new\; \labelL\; \procQi}{}
\]
Moreover, $\new_{\labelL} \procPi\; \relR_{\stateS}\; \new_{\labelL} \procQi$ as 
required.
\end{proof}
}%END LONG VERSION

\subsection{Prefixing}
We assume that $\proc$ is closed under operator $\pref: \labels \times \proc \to \proc$, 
whose semantics is standard.
\longversion{
\[
\irule{}{\labelL \pref \procP \trans[\labelL] \procP}{}
\]}
We proceed similarly to the case of $\new$: we treat the prefix operator
$\pref$ as an $\labels$-indexed family of unary operators $\pref_{\labelL}$. 

\begin{example}
Let $\procP[0]$ and $\procQ[0]$ be as in \Cref{ex:restriction}. 
Since $\labelB \pref \procP[0]$ and $\labelB \pref \procQ[0]$ can only reduce with
a $\labelB$ move to, respectively, $\procP[0]$ and $\procQ[0]$, their distance 
$\metric(\labelB \pref \procP[0],\labelB \pref \procQ[0])$ should reduce to
$\metric(\procP[0],\procQ[0]) = \stateS[0]$. Moreover, after performing an $\labelA$ action, $\metric(\labelB \pref \procP[0],\labelB \pref \procQ[0])$ should reduce to $\mbot$.
%We have that
%$\metric(\labelB \pref \procP[0],\labelB \pref \procQ[0]) = \stateSi[0]$.
%Indeed, $\stateSi[0] \trans[\labelB] \stateS[0] = \metric(\procP[0],\procQ[0])$
%and $\stateSi[0] \trans[\labelB] \mbot$. 

%\vspace{-0.2cm}
\begin{center}
	\includegraphics{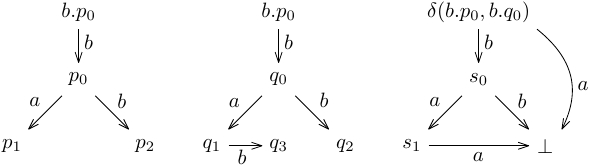}
	\commentout{
\begin{tikzpicture}[node distance = 2cm, on grid,auto]
\node (q00)[state, above = of q0][label =above:{$\labelB \pref \procP[0]$}] {};
\node (q0) [state][label =left:{$\procP[0]$}] {};
\node (q1) [state, below left = of q0][label =above:{$\procP[1]$}] {};
\node (q2) [state, below right = of q0][label =above:{$\procP[2]$}] {};
\node (p00)[state, above = of p0][label =above:{$\labelB \pref \procQ[0]$}] {};
\node (p0) [state,right =4.5cm of q0][label =left:{$\procQ[0]$}] {};
\node (p1) [state, below left = of p0][label =above:{$\procQ[1]$}] {};
\node (p2) [state, below right = of p0][label =above:{$\procQ[2]$}] {};
\node (p3) [state, below = of p0][label =above:{$\procQ[3]$}] {};
\node (s00)[state, above = of s0][label =above:{$\metric(\labelB \pref \procP[0],
\labelB \pref \procQ[0])$}] {};
\node (s0) [state,right =3.5cm of p0][label =left:{$\stateS[0]$}] {};
\node (s1) [state, below = of s0][label =left:{$\stateS[1]$}] {};
\node (s2) [state, below right = of s0][label =right:{$\mbot$}] {};
\path [->]
(q00) edge [] node {$\labelB$}   (q0)
(q0) edge [] node {$\labelA$}   (q1)
(q0) edge [] node {$\labelB$}   (q2)
(p00) edge [] node {$\labelB$}   (p0)
(p0) edge [] node {$\labelA$}   (p1)
(p0) edge [] node {$\labelB$}   (p2)
(p1) edge [] node {$\labelB$}   (p3)
(s00) edge [] node {$\labelB$}   (s0)
(s00) edge [bend left] node {$\labelA$}   (s2)
(s0) edge [] node {$\labelA$}   (s1)
(s0) edge [] node {$\labelB$}   (s2)
(s1) edge [] node {$\labelA$}   (s2)
;
\end{tikzpicture}}
\end{center}
%\vspace{-0.8cm}
\qed
\end{example}

In our contextual setting, prefixing of processes can change the distance, 
and the new distance may be incomparable to the original one. 
Therefore properties like non-expansiveness
do not hold in general for $\hat{\pref_{\labelL}}$. Among the compositionality 
properties appeared in literature, uniform continuity \cite{GeblerLT15} seems 
appropriate for prefixing. Uniform continuity holds when for all 
$\stateS[\epsilon] \greatThanStrict[\automatonV] \mbot[\automatonV]$ there is 
$\stateS[\delta] \greatThanStrict[\automatonV] \mbot[\automatonV]$ such that
$\hat{\pref_{\labelL}}(\procP,\stateS[\delta]) \lessThan[\automatonV] 
\stateS[\epsilon]$. Such condition is too strong: for instance if 
$\stateS[\epsilon] \trans[\labelL] \mbot[\automatonV]$ the only option is to take 
$\stateS[\delta] = \mbot[\automatonV]$, %which does not satisfy the requirement 
hence $\stateS[\delta] \not\greatThanStrict[\automatonV] \mbot[\automatonV]$.
   
For this reason, we need a stronger property for $\stateS[\epsilon]$, namely that
the meet of the set of $\labelL$ reducts of $\stateS[\epsilon]$ is strictly greater 
than $\mbot[\automatonV]$ and its immediate value is lower than that of 
$\stateS[\epsilon]$.
%\[
%\glb \setcomp{\stateS}{\stateS[\epsilon] \trans[\labelL] \stateS} \greatThanStrict \mbot
%\]

\longversion{%BEGIN LONG VERSION
We start with the following auxiliary lemma.
\begin{lemma}\label{lem:aux:prefix}
For all $\stateS[\epsilon]$, if 
$\stateS[\delta] = \glb[\automatonV] \setcomp{\stateS}{\stateS[\epsilon] \trans[\labelL] \stateS}$ 
is such that %$\stateS[\delta] \greatThanStrict \mbot$ and 
$\val \stateS[\delta] \lessThan[\quantaleQ] \val \stateS[\epsilon]$, then for all
$\procP: \hat{\pref_{\labelL}}(\procP,\stateS[\delta]) \lessThan[\automatonV] \stateS[\epsilon]$.
\end{lemma}
\begin{proof}
Let $\stateS[\epsilon]$ and $\stateS[\delta]$ be as in the statement.
It suffice to show that $\procP \sim_{\stateS[\delta]}\procQ$ implies
$\labelL \pref \procP \sim_{\stateS[\epsilon]} \labelL \pref \procQ$. So, suppose
$\procP \sim_{\stateS[\delta]}\procQ$. For condition 
$\immMetric(\labelL \pref \procP,\labelL \pref \procQ) \lessThan[\quantaleQ] 
\val \stateS[\epsilon]$,  by \cref{eq:comp-imm} we have that 
$\immMetric(\labelL \pref \procP, \labelL \pref \procQ) \lessThan[\quantaleQ]
\immMetric(\procP,\procQ)$. The thesis then follows since 
$\immMetric(\procP,\procQ) \lessThan[\quantaleQ] \val \stateS[\delta]$
and $\val \stateS[\delta] \lessThan[\quantaleQ] \val \stateS[\epsilon]$.
For condition 2, suppose 
$\stateS[\epsilon] \trans[\labelLi] \stateSi[\epsilon]$.
If $\labelLi \neq \labelL$ the thesis holds trivially since neither $\labelL \pref \procP$
nor $\labelL \pref \procQ$ can fire a $\labelLi$ transition. Instead, if $\labelLi = \labelL$,
we have that the only $\labelLi$ transitions are, respectively, $\labelL \pref \procP 
\trans[\labelL] \procP$ and $\labelL \pref \procQ \trans[\labelL] \procQ$. Since
$\stateS[\delta] \lessThan[\automatonV] \stateSi[\epsilon]$ by definition and $\procP 
\sim_{\stateS[\delta]}\procQ$ by assumption, we have the thesis.
\end{proof}
}%END LONG VERSION

\begin{proposition}
For all $\stateS[\epsilon] \greatThanStrict[\automatonV] \mbot[\automatonV]$ 
such that $\stateS[\labelL] = \glb[\automatonV] \setcomp{\stateS}{\stateS[\epsilon] 
\trans[\labelL] \stateS}\greatThanStrict \mbot$ and 
$\val \stateS[\labelL] \lessThan[\quantaleQ] \val \stateS[\epsilon]$, there is 
$\stateS[\delta] \greatThanStrict[\automatonV] \mbot[\automatonV]$ such that
$\hat{\pref_{\labelL}}(\procP,\stateS[\delta]) \lessThan[\automatonV] \stateS[\epsilon]$.
\end{proposition}
\longversion{%BEGIN LONG VERSION
\begin{proof}
Set $\stateS[\delta] = \stateS[\labelL]$. The thesis follows by \cref{lem:aux:prefix}.
\end{proof}
}%END LONG VERSION

\subsection{Non-deterministic Sum}
We assume that $\proc$  is closed under binary operator $\esum$, whose semantics is again standard.
\longversion{
\[
\irule{\procP \trans[\labelL]\procPi }
{\procP \esum \procQ \trans[\labelL]\procPi}{}
\qquad
\irule{\procQ \trans[\labelL] \procQi}
{\procP \esum \procQ \trans[\labelL] \procQi}{}
\]}

\begin{example}\label{ex:sum}
Let $\procP[0], \procQ[0]$ and $\stateS[0]$ be as in \Cref{ex:restriction}, and $\procR[0]$ as
in the picture below. We have that 
$\metric(\procP[0] \esum \procR[0],\procQ[0] \esum \procR[0]) = \stateS[0]$: it reduces to $\mbot$
after a $\labelB$ move (both processes indeed terminate after a $\labelB$ action).
An $\labelA$ action instead leads to a state that can only perform a $\labelA$ action
towards $\mbot$. This is because $\procQ[0] \esum \procR[0]$ can reduce to $\procQ[1]$
with a $\labelA$ move, while $\procP[0] \esum \procR[0]$ cannot match that action exactly:
it can reduce to $\procP[1]$ or $\procR[1]$, that are not bisimilar to $\procQ[1]$.
%\vspace{-0.2cm}
\begin{center}
	\includegraphics{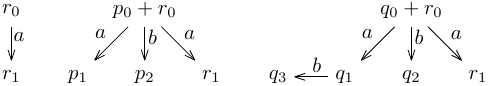}
	\commentout{
\begin{tikzpicture}[node distance = 2cm, on grid,auto]
\node (r0) [state][label =above:{$\procR[0]$}] {};
\node (r1) [state, below = of r0][label =below:{$\procR[1]$}] {};
\node (q0) [state,right =3cm of r0][label =above:{$\procP[0] \esum \procR[0]$}] {};
\node (q1) [state, below left = of q0][label =above:{$\procP[1]$}] {};
\node (q2) [state, below = of q0][label =below:{$\procP[2]$}] {};
\node (q3) [state, below right = of q0][label =below:{$\procR[1]$}] {};
\node (p0) [state,right =4.9cm of q0][label =above:{$\procQ[0]\esum \procR[0]$}] {};
\node (p1) [state, left = of p0][label =above:{$\procQ[1]$}] {};
\node (p2) [state, below = of p0][label =below:{$\procQ[2]$}] {};
\node (p3) [state, below = of p1][label =below:{$\procQ[3]$}] {};
\node (p4) [state, below right= of p0][label =below:{$\procR[1]$}] {};
%\node (s0) [state,right =3.5cm of p0][label =above:{$\stateS[0]$}] {};
%\node (s1) [state, below = of s0][label =left:{$\stateS[1]$}] {};
%\node (s2) [state, below right = of s0][label =right:{$\mbot$}] {};
\path [->]
(r0) edge [] node {$\labelA$}   (r1)
(q0) edge [] node {$\labelA$}   (q1)
(q0) edge [] node {$\labelB$}   (q2)
(q0) edge [] node {$\labelA$}   (q3)
(p0) edge [] node {$\labelA$}   (p1)
(p0) edge [] node {$\labelB$}   (p2)
(p0) edge [] node {$\labelA$}   (p4)
(p1) edge [] node {$\labelB$}   (p3)
%(s0) edge [] node {$\labelA$}   (s1)
%(s0) edge [] node {$\labelB$}   (s2)
%(s1) edge [] node {$\labelA$}   (s2)
;
\end{tikzpicture}}
\end{center}
%\vspace{-0.8cm}
\qed
\end{example}

Intuitively, the non-deterministic sum of two precesses can behave as the former 
process or as the latter (but not as both). Therefore we can expect that the distance
between two sums is bounded by the join of the distances of the components. This
is however not always the case, as the immediate distance is not necessarily
non-expansive. The sum operator $\mplus[\automatonV]$ from 
\cref{def:MLTS}, instead, turns out to be sufficient for our purposes.
Proposition below indeed shows that $\hat{\esum}$ is non-extensive.
\begin{proposition}
For every $\procP[1],\procP[2],\stateS[1],\stateS[2]:$ it holds that 
$\hat{\esum}(\procP[1],\procP[2],\stateS[1],\stateS[2]) \lessThan[\automatonV]
\stateS[1] \mplus[\automatonV] \stateS[2]$.
\end{proposition}
%BEGIN LONG VERSION
\longversion{
\begin{proof}
We start by showing that: \[\metric(\procP[1] \esum \procP[2],
\procQ[1] \esum \procQ[2]) \lessThan[\automatonV] \stateS[1] \mplus \stateS[2]\] 
Where:
\[
\metric(\procP[1],\procQ[1]) \lessThan[\automatonV] \stateS[1] \qquad
\metric(\procP[2],\procQ[2]) \lessThan[\automatonV] \stateS[2] 
\]
We rely on \cref{lem:metric-bisim}, and we prove:
\[
\procP[1] \esum \procP[2] \sim_{\stateS[1] \mplus[\automatonV] \stateS[2]} 
\procQ[1] \esum \procQ[2]
\]
Condition $\immMetric(\procP[1] \esum \procP[2],\procQ[1] \esum \procQ[2]
\lessThan[\quantaleQ] \val \stateS[1] \mplus[\automatonV] \stateS[2]$ follows directly from
thr assumption \cref{eq:comp-imm}. So, suppose $\stateS[1] \mplus[\automatonV] \stateS[2]
\trans[\labelL] \stateS$. It must be $\stateS = \stateSi[1] \mplus \stateSi[2]$
for some $\stateSi[1],\stateSi[2]$ such that $\stateS[1] \trans[\labelL] \stateSi[1]$ 
and $\stateS[2] \trans[\labelL] \stateSi[2]$. Suppose
$\procP[1] \esum \procP[2] \trans[\labelL] \procP$. 
By inversion on the operational semantics, we have that
$\procP[1] \trans[\labelL] \procP$ or $\procP[2] \trans[\labelL] \procP$. We show
only the former case. Since $\stateS[1] \trans[\labelL] \stateSi[1]$, we have
that $\procQ[1] \trans[\labelL] \procQi[1]$ for some $\procQi[1]$ such that
$\procPi[1] \sim_{\stateSi[1]} \procQi[1]$. Since $\stateSi[1] \lessThan[\automatonV] 
\stateSi[1] \mplus[\automatonV] \stateSi[2]$, we have that 
$\procPi[1] \sim_{\stateSi[1] \mplus \stateSi[2]} \procQi[1]$, as required.
The case for $\procQ[1] \esum \procQ[2]$ moves is similar.
\end{proof}
}%END LONG VERSION

\subsection{Parallel Composition}\label{compositionality:paral}
We assume $\proc$ to be closed under the binary operator $\paral$, whose semantics is
defined below:
\[
\irule{\procP \trans[\labelL]\procPi }
{\procP \paral \procQ \trans[\labelL]\procPi \paral \procQ}{}
\qquad
\irule{\procP \trans[{\labelL}]\procPi \quad \procQ \trans[{\labelL}] \procQi}
{\procP \paral \procQ \trans[\labelL]\procPi \paral \procQi}{}
\qquad
\irule{\procQ \trans[\labelL] \procQi}
{\procP \paral \procQ \trans[\labelL]\procP \paral \procQi}{}
\]
The notion of synchronisation considered in this paper is the one pioneered in CSP \cite{Hoare78,Glabbeek97}.
This choice is motivated by the fact that, in comparison with CCS-like communication \cite{Milner80} 
(which requires dual actions to synchronise resulting in an invisible $\tau$-action), CSP
notion does not change the label: this simplifies the technical development and enables 
stronger compositionality properties. Most of the works on compositionality of metrics for parallel composition
we are aware of use CSP synchronisation, e.g. \cite{BacciBLM13,GeblerLT16,GeblerLT15}.

\begin{example}\label{ex:paral}
Let $\procP[0]$ and $\procQ[0]$ be as in \Cref{ex:restriction}, and $\procR[0]$ as in
\Cref{ex:sum}. We have that $\metric(\procP[0] \paral \procR[0],\procQ[0] \paral \procR[0])$
is as the figure below. Indeed, $\procP[0] \paral \procR[0]$ and $\procQ[0] \paral \procR[0]$
necessarily reduce to bisimilar states after a $\labelB$ action: therefore their distance
$\labelB$-reduces to $\mbot$. The situation for $\labelA$ actions is more involved, due
the the presence of several $\labelA$-reducts for both processes. So, consider the transition
$\procP[0] \paral \procR[0] \trans[\labelA] \procP[1] \paral \procR[0]$. We need to find
the matching move of $\procQ[0]\paral\procR[0]$ that minimises the distance between the 
reducts. So, consider the transition $\procQ[0] \paral \procR[0] \trans[\labelA] \procQ[1] 
\paral \procR[1]$. Since $\procP[1] \paral \procR[0]$ can only perform $\labelA$ actions
while $\procQ[1] \paral \procR[1]$ only $\labelB$ ones, we have that 
$\metric(\procP[1] \paral \procR[0],\procQ[1] \paral \procR[1]) = \mtop$. If we instead
consider transition $\procQ[0] \paral \procR[0] \trans[\labelA] \procQ[1] \paral \procR[0]$,
we have that $\metric(\procP[1] \paral \procR[0],\procQ[1] \paral \procR[0]) = \stateSi[1]$.
Indeed, $\procQ[1] \paral \procR[0] \trans[\labelB]$ while $\procP[1] \paral \procR[0]$
does not: hence $\stateSi[1] \not\trans[\labelB]$. Moreover, $\stateSi[1] \trans[\labelA]
\stateSi[2]$. The only $\labelA$-reducts of 
$\procP[1] \paral \procR[0]$ and $\procQ[1] \paral \procR[0]$ are, respectively,
$\procP[1] \paral \procR[1]$ and $\procQ[1] \paral \procR[1]$.
It is easy to verify that $\metric(\procP[1] \paral \procR[1],\procQ[1] \paral \procR[1])
= \stateSi[2]$. The last possible matching choice is 
$\procQ[0] \paral \procR[0] \trans[\labelA] \procQ[0] \paral \procR[1]$, for which
we have that $\metric(\procP[1] \paral \procR[0],\procQ[0] \paral \procR[1]) = \stateSi[1]$:
the argument is similar to the previous case.
All the other starting $\labelA$-moves of $\procP[0] \paral \procR[0]$, and those of 
$\procQ[0] \paral \procR[0]$, have matching moves leading to distances greater or equal
than $\stateSi[1]$.
%\vspace{-0.2cm}
\begin{center}
	\includegraphics{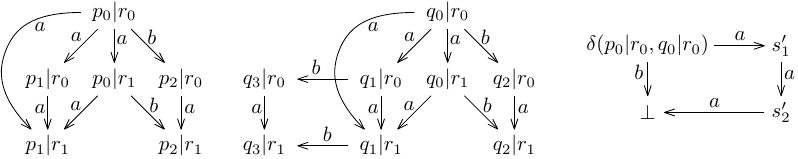}
\end{center}
%\vspace{-0.8cm}
\qed
\end{example}

Parallel composition does not enjoy strong compositionality properties. Indeed in general
$\hat{\paral}(\procP[1],\procP[2],\stateS[1],\stateS[2])$ is  related neither to $\stateS[1]$ nor to
$\stateS[2]$, and even $\hat{\paral}(\procP[1],\procP[2],\stateS[1],\mbot[\automatonV])$ is not
related to $\stateS[1]$. Consider for instance the case where $\procP[2]$ ``consumes'' a $\stateS[1]$
move.

%as shown by the following example.

%\begin{example}
%Let $\automatonV$ be as in \Cref{sec:larsen}, with 
%$\labels = \{\labelA,\labelB,\labelC\}$. Let $\procP[1],\procP[2],\procQ[2] \in \proc$ 
%be such that their only transitions are: 
%\begin{itemize}
%\item $\procP[1] \trans[\labelA] \procP[1]$;
%\item $\procP[2] \trans[\labelB] 0$;
%\item $\procQ[2]\trans[\labelB] 0$.
%\end{itemize}
%Where $0 \in \proc$ is the inactive process.
%Let $\stateS,\stateSi,\stateSii \in \states$ be such that their only transitions are:
%\begin{itemize}
%\item
%\stateS \trans[\labelA] \mbot;
%\item
%\stateSi \trans
%\end{itemize}
%\end{example}

However, our metric domain $\automatonV$ contains ``contextual'' information.
We exploit this fact to show that a nice compositionality property, similar to
non-extensivity \cite{GeblerLT15}, holds when the context and the distance are ``compatible''.
A formal definition of compatibility follows. 

\begin{definition}\label{def:comp-relation}
A relation $\relR \subseteq \states \times \proc$ is a compatibility relation if,
whenever $\stateS\; \relR\; \procP$:
\begin{enumerate}
\item
$\stateS \trans[\labelL] \stateSi \implies  \stateSi\; \relR\; \procP$;
\item
$\stateS \trans[\labelL] \stateSi \;\text{and}\;\procP \trans[\labelL] \procPi \implies
\stateSi\; \relR\; \procPi \;\text{and}\; \stateS \lessThan[\automatonV] \stateSi$.
\end{enumerate}
We say that $\stateS$ is $\procP$-compatible iff $\stateS\; \relR\; \procP$ for some
compatibility relation $\relR$.
\end{definition}

\begin{example}
Consider again $\procP[0]$, $\stateS[0]$, $\stateS[1]$ from \cref{ex:restriction}
and $\metric(\procP[0] \paral \procR[0],\procQ[0] \paral \procR[0])$ of \cref{ex:paral}. 
We have that 
$\stateS[0]$ is not $\procP[0]$-compatible as Condition 2 from \cref{def:comp-relation} is 
violated: $\stateS[0] \trans[\labelA] \stateS[1]$ and $\procP[0] \trans[\labelA] \procP[1]$ 
but $\stateS[0] \not\lessThan[\automatonV_0] \stateS[1]$. Instead, $\stateSii[0]$ below
is $\procP[0]$-compatible: it follows from the facts that $\stateSii[0]$ 
necessarily reduces to a greater or equal state, $\procP[0]$ reduces to terminated states, 
which are vacuously compatible with every distance. Note that 
$\metric(\procP[0],\procQ[0]) = \stateS[0] \lessThan[\automatonV_0] \stateSii[0]$ and 
$\metric(\procP[0] \paral \procR[0],\procQ[0] \paral \procR[0])
\lessThan[\automatonV_0] \stateSii[0]$. The second inclusion follows
from the first by \cref{lem:composition:paral}.
%\vspace{-0.2cm}
\begin{center}
\includegraphics{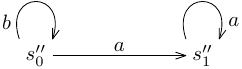}
\commentout{
\begin{tikzpicture}[node distance = 2cm, on grid,auto]
\node (s0) [state][label =left:{$\stateSii[0]$}] {};
\node (s1) [state, right  = of s1][label =right:{$\stateSii[1]$}] {};
%\node (s2) [state, right = of s0][label =right:{$\mbot$}] {};
\path [->]
(s0) edge [bend left] node {$\labelA$}   (s1)
(s0) edge [loop above] node {$\labelB$}   (s0)
(s1) edge [loop above] node {$\labelA$}   (s1)
;
\end{tikzpicture}}
\end{center}
%\vspace{-0.7cm}
\end{example}

%\begin{definition}
%A relation $\relR \subseteq \states \times \proc \times \states$ is a $\Delta$-compatibility
%relation if, whenever $(\stateS,\procP,\stateS[\Delta]) \in \relR$:
%\begin{enumerate}
%\item
%$\stateS \trans[\labelL] \stateSi \implies  (\stateSi,\procP,\stateS[\Delta]) \in \relR$;
%\item
%$\stateS \trans[\labelL] \stateSi \;\text{and}\;\procP \trans[\labelL] \procPi 
%\land \stateS[\Delta] \trans[\labelL] \stateSi[\Delta] \implies
%(\stateSi,\procP,\stateS[\Delta]) \in \relR \;\text{and}\; \stateS \lessThan \stateSi$;
%\item
%$\stateS \trans[\labelL] \implies \procP \trans[\labelL] \;\text{and}\; 
%\stateS[\Delta] \trans[\labelL]$.
%\end{enumerate}
%We say that $\stateS$ is $\procP,\stateS[\Delta]$-compatible iff 
%$(\stateSi,\procP,\stateS[\Delta]) \in \relR$ for some
%$\Delta$-compatibility relation $\relR$.
%\end{definition}
%
%\begin{lemma}
%If $\stateS$ is $\procP,\stateS[\Delta]$-compatible and $\procQ$ is such that 
%$\metric(\procP,\procQ) \lessThan \stateS[\Delta]$, then $\stateS$ is 
%$\procQ$-compatible. 
%\end{lemma}
%\begin{proof}
%Sketch on paper.
%\end{proof}

\begin{proposition}\label{lem:composition:paral}
If $\stateS[1]$ is $\procP[2]$-compatible and 
$\stateS[2]$ is $\procP[1]$-compatible, then
$\hat{\paral}(\procP[1],\procP[2],\stateS[1],\stateS[2]) \lessThan[\automatonV] \stateS[1] 
\mplus[\automatonV] \stateS[2]$.
\end{proposition}
%BEGIN LONG VERSION
\longversion{
\begin{proof}
It suffice to prove that $\relR_{\stateS[1] \mplus[\automatonV] \stateS[2]} = 
\{(\procP[1] \paral \procP[2],\procQ[1] \paral \procQ[2]) \mid
\metric(\procP[1],\procQ[1]) \lessThan[\automatonV] \stateS[1],\metric(\procP[2],\procQ[2])
\lessThan \stateS[2], \stateS[1]$ is $\procP[2]$-compatible and
$\stateS[2]$ is $\procP[1]$-compatible$\}$ is a parametrized bisimulation.
So, let $\procP[1] \paral \procP[2] \relR_{\stateS[1]\mplus\stateS[2]} 
\procQ[1] \paral \procQ[2]$. Condition 
$\immMetric(\procP[1] \paral \procP[2],\procQ[1] \paral \procQ[2]) \lessThan[\quantaleQ] \val
\stateS[1] \mplus[\automatonV] \stateS[2]$ follows immediately by \cref{eq:comp-imm}.
For condition 2, suppose $\stateS[1] \mplus[\automatonV] \stateS[2] \trans[\labelL] \stateS$. By 
inversion, it must be $\stateS = \stateSi[1] \mplus[\automatonV] \stateSi[2]$ for some $\stateSi[1],
\stateSi[2]$ such that $\stateS[1] \trans[\labelL] \stateSi[1]$ and
$\stateS[2] \trans[\labelL] \stateSi[2]$. So, suppose $\procP[1] \paral \procP[2] 
\trans[\labelL]$. We proceed by cases on the rule used:
\begin{itemize}
\item
\[
\irule{\procP[1] \trans[\labelL]\procPi[1]}
{\procP[1] \paral \procP[2] \trans[\labelL] \procPi[1] \paral \procP[2]}{}
\]
Then, $\procQ[1] \trans[\labelL]\procQi[1]$ for some $\procQi[1]$ such that
$\metric(\procPi[1],\procQi[1]) \lessThan[\automatonV] \stateSi[1]$. Then:
\[
\irule{\procQ[1] \trans[\labelL]\procQi[1]}
{\procQ[1] \paral \procQ[2] \trans[\labelL] \procQi[1] \paral \procQ[2]}{}
\]
By item 1 of \cref{def:comp-relation}, we have that $\stateSi[1]$ is  
$\procP[2]$-compatible. By item 2 of \cref{def:comp-relation}, 
we have that $\stateSi[2]$ is $\procPi[1]$-compatible and 
$\stateS[2] \lessThan \stateSi[2]$. Hence $\metric(\procP[2],\procQ[2])
\lessThan[\automatonV] \stateS[2]$. Therefore:
\[
\procPi[1] \paral \procP[2]\;\relR_{\stateSi[1] \mplus[\automatonV] \stateSi[2]}\;
\procQi[1] \paral \procQ[2]
\]
\item
\[
\irule{\procP[1] \trans[\labelL]\procPi[1] \quad \procP[2] \trans[\labelL]\procPi[2]}
{\procP[1] \paral \procP[2] \trans[\labelL] \procPi[1] \paral \procPi[2]}{}
\]
Then, $\procQ[1] \trans[\labelL]\procQi[1]$ and $\procQ[2] \trans[\labelL]\procQi[2]$ 
for some $\procQi[1],\procQi[2]$ such that $\metric(\procPi[1],\procQi[1]) \lessThan[\automatonV] 
\stateSi[1]$ and $\metric(\procPi[2],\procQi[2]) \lessThan[\automatonV] \stateSi[2]$. Then:
\[
\irule{\procQ[1] \trans[\labelL]\procQi[1] \quad \procQ[2] \trans[\labelL]\procQi[2]}
{\procP[1] \paral \procP[2] \trans[\labelL] \procQi[1] \paral \procQi[2]}{}
\]
By item 2 of \cref{def:comp-relation}, we have that $\stateSi[1]$ is  
$\procPi[2]$-compatible and $\stateSi[2]$ is $\procPi[1]$-compatible.
%$\stateS[1] \lessThan \stateSi[1]$ and $\stateS[2] \lessThan \stateSi[2]$. 
%Hence $\metric(\procP[1],\procQ[1]) \lessThan \stateS[1]$ and 
%$\metric(\procP[2],\procQ[2])\lessThan \stateS[2]$.
Therefore:
\[
\procPi[1] \paral \procPi[2]\;\relR_{\stateSi[1] \mplus[\automatonV] \stateSi[2]}\;
\procQi[1] \paral \procQi[2]
\]
\end{itemize}
The case for the last rule is similar, as is the case for $\procQ[1]\paral\procQ[2]$
moves.
\end{proof}
}%END LONG VERSION

\subsection{Replication}
We assume that $\proc$  is closed both under operator $\paral$ (as defined in 
\cref{compositionality:paral}) and under $\bang: \proc \to \proc$, 
whose semantics is standard.
\longversion{
\[
\irule{\procP \trans[\labelL]\procPi }
{\bang \procP \trans[\labelL]\procPi \paral \bang \procP}{}
\]}
In general, replication has bad compositionality properties: since it allows 
infinite behaviour,
even a small distance in the parameter can get amplified to a much larger value.
However, we show that $\hat{\bang}$ is not expansive under the assumption that
the parameter $\stateS$ always reduces to a larger or equal value and 
$\mplus[\quantaleQ]$ is idempotent. Such condition is
of course quite strong, but it holds for instance when interpreting bisimilarity as a 
contextual bisimulation metric (see \cref{ex:bisimilarity}).

\begin{example}
We have that $\bang \procP[0]$ and $\bang \procQ[0]$ can both fire a $\labelA$ or $\labelB$
action and reduce to a process with the same behaviour (the simplest state with 
this
property is drawn in the figure). Therefore, the distance 
$\metric(\bang\procP[0],\bang\procQ[0]) = \mbot[\automatonV]$. In general, however,
the distance among processes is not preserved by replication, as shown below:
%\vspace{-0.2cm}
\begin{center}
	\includegraphics{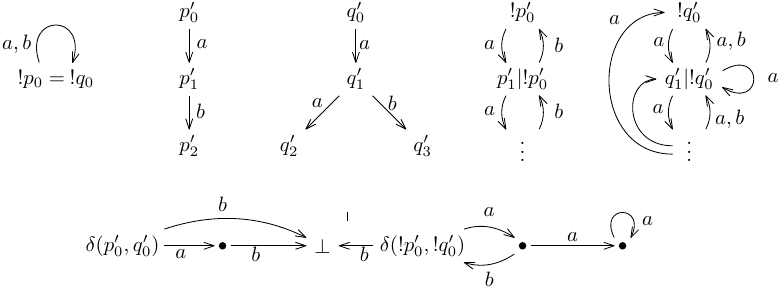}
\end{center}
%\vspace{-0.9cm}
\qed
\end{example}

\begin{definition}
We define $\mathbf{Inc}$, the set of increasing states, as the largest set 
$\statesi \subseteq \states$ such that, whenever $\stateS \in \statesi$
and $\stateS \trans[\labelL] \stateSi: \stateS \lessThan[\automatonV] \stateSi$ and
$\stateSi \in \statesi$.
\end{definition}
 
\begin{proposition}
If $\stateS$ is increasing and $\mplus[\quantaleQ]$ is idempotent, then
$\hat{\bang}(\procP,\stateS) \lessThan[\automatonV] \stateS$.
\end{proposition}
%BEGIN LONG VERSION
\longversion{
\begin{proof}
We prove that, provided $\mplus[\quantaleQ]$ is idempotent, $\relR_{\stateS}$
defined below is a parametrized bisimulation.
\[
\relR_{\stateS} = 
\setcomp{(\procP[1] \paral (\hdots \paral (\procP[n] \paral \bang \procP)),
\procQ[1] \paral (\hdots \paral (\procQ[n] \paral \bang \procQ))}{n \geq 0, \stateS \in 
\;Dec, \forall 1 \leq i \leq n: \metric(\procP[i],\procQ[i]) \lessThan \stateS,
\metric(\procP,\procQ) \lessThan \stateS}
\]
So, let $\hat{\procP} \relR_{\stateS} \hat{\procQ}$.
For condition 1 we have that 
$\immMetric(\hat{\procP},\hat{\procQ}) \lessThan[\quantaleQ]
(\Sigma_{1 \leq i \leq n} \immMetric(\procP[i],\procQ[j])) \mplus[\quantaleQ]
\immMetric(\procP,\procQ) \lessThan \val \stateS$. The first inequality follows from
\cref{eq:comp-imm}, the second from idempotency.
For condition 2, suppose $\stateS \trans[\labelL] \stateSi$ and $\hat{\procP} 
\trans[\labelL] \hat{\procPi}$. First notice that $\stateSi \in Inc$ and $\stateS \lessThan[\automatonV] 
\stateSi$. By inversion and a routine induction on $n$ (omitted) we can conclude that
$\hat{\procPi}$ is in one of the following shapes:
\begin{enumerate}
\item
\[
\hat{\procPi} = \procPi[1] \paral (\hdots \paral (\procPi[n] \paral \bang \procP))
\]
where there is non-empty $I \subseteq \{1,\hdots,n\}$ such that
$i \in I \implies \procP[i] \trans[\labelL] \procPi[i]$ and 
$i \not\in I \implies \procP[i] = \procPi[i]$;
\item
\[
\hat{\procPi} = \procPi[1] \paral (\hdots \paral (\procPi[n] \paral 
(\procPi \paral \bang \procP)))
\]
where there is (possibly empty) $I \subseteq \{1,\hdots,n\}$ such that
$\procP \trans[\labelL] \procPi$
$i \in I \implies \procP[i] \trans[\labelL] \procPi[i]$ and 
$i \not\in I \implies \procP[i] = \procPi[i]$;
\end{enumerate}
We show only case 2, which is slightly more involved. So, we have that, for all $i \in I:
\procQ[i] \trans[\labelL] \procQi[i]$ for some $\procQi[i]$ such that 
$\metric(\procPi[1],\procQi[1]) \lessThan[\automatonV] \stateSi$. Furthermore, there is $\procQi$
such that $\procQ \trans[\labelL] \procQi$ and $\metric(\procPi,\procQi) \lessThan[\automatonV] 
\stateSi$. So, let $\hat{\procQi} = \procQi[1] \paral (\hdots \paral (\procQi[n] \paral 
(\procQi \paral \bang \procQ)))$, where $i \not\in I \implies \procQ[i] = \procQi[i]$.
We have that $\hat{\procQ} \trans[\labelL] \hat{\procQi}$ and 
$\hat{\procPi} \relR_{\stateS} \hat{\procQi}$, as required.
\end{proof}
}%END LONG VERSION

\section{Related Work \& Conclusion}
%!TeX root=main.tex
%\subparagraph*{Contextual bisimulation relations.}
Quite a few works in the literature study context dependent relations. The 
closest to our work is the already mentioned study about environment
paremetrized bisimilarity \cite{Larsen87}. Our definition of CBM is similar to theirs, where the main differences are that we also consider
quantitative aspects and that we explicitly work with a metric. The same work also provides an interesting logical characterisation 
of their relation in terms of Hennessy-Milner logic, but does not study 
compositionality. Since environment parametrized bisimilarity can be 
embedded into our framework, our compositionality results also hold for \cite{Larsen87}. 
 A closely related line of research  \cite{BeoharKKS20,HulsbuschK12,Hulsbusch0KS20} (non-exhaustive list) studies conditional bisimulations in an abstract categorical framework,
where conditions are used to make assumptions on the environment. In particular, \cite{Hulsbusch0KS20} introduces a notion of conditional bisimilarity for reactive systems and shows that conditional bisimilarity is a congruence.
In \cite{HennessyL95}, an early and a late notion of symbolic bisimilarity for value passing processes are introduced, where actual values
are symbolically represented with boolean expressions with free variables. Symbolic bisimilarities are parametric w.r.t. a predicate
that, in a sense, allows to make assumptions on the values that the context can send. Our notion of contextuality instead restricts the 
choices of the environment, and we do not consider explicit value passing.

%\subparagraph*{Compositionality in behavioural metrics.}
Compositionality of behavioural metrics has been studied in the probabilistic 
setting \cite{ChatzikokolakisGPX14,DesharnaisGJP04}. 
In \cite{BacciBLM13}, it has been shown that
parallel composition is non-extensive. We remark that our notion of parallel composition is slightly more general than
the one considered in \cite{BacciBLM13}, as in there processes necessarily synchronise on common actions. The work
\cite{GeblerLT16} studies compositionality for quite a few process algebraic operators, showing e.g. that non-deterministic sum is 
non-expansive, while parallel composition is non-extensive. The bang operator is shown Lipschitz continuous for the discounted metric, 
while not even uniformly continuous w.r.t. the non-discounted one. \cite{GeblerT18} introduces 
structural operational semantics formats that
guarantee compositionality of operators. Basically, compositionality depends on how many parameters of the operator are copied
from the source to the destination of the rules, weighted by probabilities and the discount factor.
%\paragraph*{Contextual metrics.}
%Several works study contextual metrics for the $\lambda$-calculus and higher order languages. \cite{LagoGY19} introduces differential logical
%relations for a simply typed $\lambda-calculus$, which enable contextual metric reasoning since the distance of two functions is by definition

%\section{Conclusion}
%!TeX root=main.tex
\subparagraph*{Concluding Remarks.}
This paper introduces a new form of metric on the states of a LTS, called 
contextual behavioural metric, which enables contextual and quantitative 
reasoning. We study compositional properties of CBMs w.r.t. some operators, 
showing that, under the assumption that the immediate metric
is non-extensive, the following hold:
restriction is non-expansive, non-deterministic sum is non-extensive,
prefixing enjoys a property slightly weaker than uniform continuity,
parallel composition is non-extensive when the distance between components is 
compatible with the context and
replication enjoys non-expansiveness under some (rather strong) assumptions on 
the underling quantale $\quantaleQ$.

\highlight{Due to the generality of CBMs, our compositionality results extend
to behavioural metrics as defined in \cref{ex:beh-bisimilarity}. For instance, since
the compatibility relation of \cref{def:comp-relation} holds trivially for the MLTS of
behavioural metrics, we have that compositionality of parallel composition only depends
on the compositionality of the immediate metric.
}

Our work is still preliminary, and indeed we are yet in the quest for an appropriate general notion of compositionality: here we tried to adapt
concepts from the probabilistic setting \cite{GeblerLT16,GeblerT18}, where 
uniform continuity is considered as the most general notion of 
compositionality. In our setting not even prefixing enjoys uniform 
continuity, which should not come as a surprise, as quantales are not 
totally ordered in general. Our compositionality results have heterogeneous side conditions. Spelling out all the compositionality results in a uniform way would come with a high price: operators for which compositionality holds without any side condition, such as restriction, would have to be treated as those for which compositionality holds only modulo appropriate (and strong) hypotheses, such as replication. An interesting future work would be to infer the side conditions directly
from SOS rules, or studying more operators or rule formats as in \cite{GeblerT18}.

Another direction of future research would be to consider calculi with value and/or channel passing like the $\pi$-calculus: since strong
bisimilarity is not a congruence in such settings, a promising approach could be a ``contextualisation'' of open-bisimilarity \cite{Sangiori93}.

\bibliography{bibliography.bib}
%\nocite{*} %da cancellare (solo per prova)

\end{document}